\documentclass{article}
\usepackage{arxiv}

\usepackage{amsmath,amsfonts,amsthm}
\usepackage{thmtools, thm-restate}
\usepackage{xr-hyper} 
\usepackage{bm}
\usepackage[title]{appendix}
\usepackage[explicit,compact]{titlesec}             
\usepackage[ruled,vlined,linesnumbered]{algorithm2e}
\usepackage[noend]{algpseudocode}
\usepackage{array}
\usepackage[caption=false,font=normalsize,labelfont=sf,textfont=sf]{subfig}
\usepackage{textcomp}
\usepackage{stfloats}
\usepackage{url}
\usepackage{verbatim}
\usepackage{graphicx}
\usepackage{cite}
\usepackage{soul} 

\usepackage{amssymb}

\newtheorem{problem}{Problem}
\newtheorem{definition}{Definition}
\newtheorem{lemma}{Lemma}
\newtheorem{proposition}{Proposition}

\newtheorem{theorem}{Theorem}
\newtheorem{corollary}{Corollary}
\newtheorem{remark}{Remark}

\usepackage[dvipsnames]{xcolor}
\newcommand{\todo}[1]{\textcolor{red}{[TODO: #1]}}

\usepackage{dsfont}

\newcommand{\e}{\mathbf{e}}
\newcommand{\x}{\mathbf{x}}
\newcommand{\xhat}{\hat{\mathbf{x}}}
\newcommand{\vx}{\mathbf{v}}
\newcommand{\vy}{\mathbf{w}}

\newcommand{\y}{\mathbf{y}}

\newcommand{\h}{\mathbf{h}}
\newcommand{\ones}{\mf{1}}

\newcommand{\R}{\mathbb{R}}
\newcommand{\mc}[1]{\mathcal{#1}}

\newcommand{\mf}[1]{\mathbf{#1}}
\newcommand{\bs}[1]{\boldsymbol{#1}}
\newcommand{\nodeset}{\mc{V}}
\newcommand{\edgeset}{\mc{E}}
\newcommand{\revcerseedgeset}{\widetilde{\edgeset}}
\newcommand{\graph}{\mc{G}}
\newcommand{\reversedgraph}{\widetilde{\graph}}

\DeclareMathOperator*{\argmax}{arg\,max}
\DeclareMathOperator*{\argmin}{arg\,min}

\renewcommand{\st}{~~\text{s.t.}~~}

\newcommand{\xreq}{\mathbf{x}^\req}
\newcommand{\xmagreq}{X^\req}
\newcommand\polyreq{\mc{P}_\req}
\newcommand{\nsp}[2]{{\mc{Q}}_{#1}^{(#2)}}

\def\RSet{\mathcal{R}}
\def\RsetInv{\mathcal{R}^{-1}}

\newcommand{\yinit}{\y_{-1}}

\newcommand{\K}{\mathcal{K}}
\newcommand{\F}{\mathcal{F}}
\newcommand{\hatK}{\hat{K}}
\newcommand{\hatF}{\hat{F}}
\newcommand{\hatcalK}{\hat{\mathcal{K}}}
\newcommand{\hatcalF}{\hat{\mathcal{F}}}

\newcommand{\convexhull}[1]{\mathrm{Conv}\big(#1\big)}
\newcommand{\conichull}[1]{\mathrm{Cone}\big(#1\big)}

\def\graph{\mc{G}}
\def\reals{\mathbb{R}}

\newcommand{\factor}[2]{\alpha_{#1}^{#2}}
\newcommand{\neighbor}[1]{\mc{N}_{#1}}

\def\dout{d}
\newcommand{\doutkey}[1]{d_{#1, \mathrm{key}}}

\def\doutmax{d_\mathrm{max}}

\newcommand{\nspfactornode}[2]{\beta_{#1}^{(#2)}}
\newcommand{\nspfactor}[1]{\beta_{#1}}

\newcommand{\keyloc}{\mathcal{V}_{\mathrm{key}}}

\newcommand{\hmax}{T}

\renewcommand{\emptyset}{\varnothing}

\def\ip{j}

\newcommand{\multfactor}{CRR}

\usepackage{mathtools}
\DeclarePairedDelimiter{\norm}{\lVert}{\rVert} 
\newcommand{\req}{\text{req}}

\newcommand{\defender}{defender}
\newcommand{\attacker}{attacker}

\let\oldnl\nl
\newcommand{\nonl}{\renewcommand{\nl}{\let\nl\oldnl}}

\renewcommand{\and}{\text{ and }}

\usepackage{yaacro}
\begin{acgroupdef}[list=acronimos]
    \acdef{DAB}{Defender-Attacker Blotto}
    \acdef{dDAB}{dynamic Defender-Attacker Blotto}
    \acdef{dDARA}{dynamic Defender-Attacker Resource Allocation}
\end{acgroupdef}

\newcommand{\thegame}{\ac{dDAB} game}

\newcommand{\figref}[1]{Figure~\ref{#1}}
\renewcommand{\algref}[1]{Algorithm~\ref{#1}}
\newcommand{\secref}[1]{Section~\ref{#1}}

\title{Dynamic Adversarial Resource Allocation:\\ The dDAB Game}

\author{
    Yue Guan \textsuperscript{1 *}\quad  
    Daigo Shishika  \textsuperscript{2 *} \quad 
    Jason R. Marden \textsuperscript{3} \\
    \textbf{Michael Dorothy  \textsuperscript{4}  \quad
    Panagiotis Tsiotras \textsuperscript{1} \quad
    Vijay Kumar \textsuperscript{5}}
    \vspace{+5pt}\\
    \textsuperscript{1}
    Georgia Institute of Technology
    \textsuperscript{2}
    George Mason University \quad
    \textsuperscript{3}
    University of California Santa Barbara \\
    \textsuperscript{4}
    United States Army Research Laboratory \quad 
    \textsuperscript{5}
    University of Pennsylvania\\
}

\begin{document}

\maketitle
\begingroup\renewcommand\thefootnote{}
\footnotetext{We gratefully acknowledge the support of ARL grant ARL DCIST CRA W911NF-17-2-0181. The views expressed in this paper are those of the authors and do not reflect the official policy or position of the United States Army, Department of Defense, or the United States Government.}
\endgroup
\begingroup\renewcommand\thefootnote{*}
\footnotetext{The first two authors contributed equally as co-first authors.}
\endgroup


\pagenumbering{arabic}
\begin{abstract}
This work introduces the dynamic Defender-Attacker Blotto (dDAB) game, extending the classical static Blotto game to a dynamic resource allocation setting over graphs. 
In the dDAB game, a \defender{} is required to maintain numerical superiority against   \attacker{} resources across a set of key nodes in a connected graph. 
The engagement unfolds as a discrete-time game, where each player reallocates its resources in turn, with resources allowed to move at most one hop per time step. 
The primary goal is to determine the necessary and sufficient amount of \defender{} resources required to guarantee sustained defense, along with the corresponding strategies. 
To address the central challenge arising from graph-constrained resource reallocation, we conduct a reachability analysis, starting with simplified settings where \attacker{} resources act as a single cohesive group. 
We then extend the framework to allow \attacker{} resources to split and merge arbitrarily,  
and construct \defender{} strategies using superposition principles.
A set-based dynamic programming algorithm is developed to compute the optimal strategies, as well as the minimum amount of \defender{} resources to ensure successful defense. 
The effectiveness of our approach is demonstrated through numerical simulations and hardware experiments on the Georgia Tech Robotarium platform.
\end{abstract}


\section{Introduction}
\label{sec:introduction}

Deploying resources (robots, sensors, or supplies) to appropriate locations at the appropriate time is a fundamental problem in multi-agent systems, often studied as the multi-robot task allocation (MRTA) problem~\cite{korsah2013comprehensive,khamis2015multi}.
In real-world settings, resource allocation or MRTA are performed in a dynamically changing environment. 
Time-varying demand is one of the major sources of dynamics, exemplified by the applications in wireless network~\cite{wireless}, ride-sharing~\cite{bei2018algorithms}, power-grid~\cite{nair2018multi}, and cloud computing~\cite{anuradha2014survey}.

In this work, we study the dynamic resource allocation problem on a graph, where nodes represent physical locations and edges represent the traversability between those locations. 
The focus is on transporting the resources effectively in the environment to satisfy demands that change dynamically. 
Instead of achieving the desired allocation instantly, we require the resources%
\footnote{We use the terms robots and resources interchangeably. The term ``player'', however, is reserved for the entity (the \defender{} or the \attacker{}) that determines the allocation of these robots / resources.} 
to \emph{traverse} through the environment.
Such consideration arises naturally when dealing with embodied agents and resources, such as robots, or autonomous vehicles. 

To stress the dynamic aspect of the problem, we consider demands that are generated by an adversary. 
Specifically, we formulate the problem as a dynamic (turn-based) game played between a blue \defender{} and a red \attacker{}. 
The objective of the \defender{} is to defend a set of key nodes by maintaining its numerical superiority over the \attacker{} resource at these nodes. 
If the \attacker{} outnumbers the \defender{} at any key node, the \defender{} loses the game. 
In that sense, the demand imposed by the attacker is a hard constraint that the \defender{} must continuously satisfy throughout the game.
Note that many other safety-critical applications with dynamic demands (e.g., resilient power grid~\cite{duan2025graph}, wildfire surveillance~\cite{julian2019distributed}, etc.) can be formulated as such a hard-constrained resource allocation problem.

\begin{figure}
    \centering
    \includegraphics[width=0.55\textwidth]{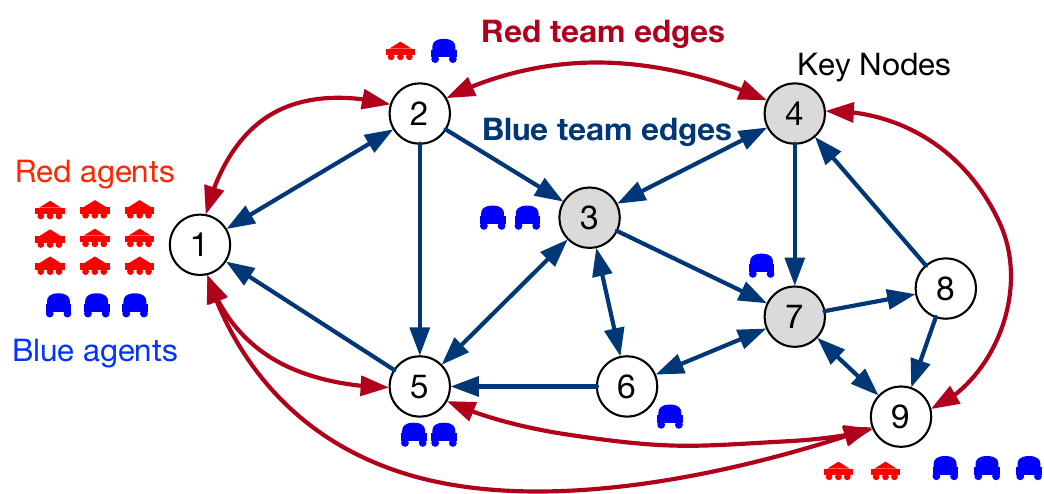}
    \caption{Illustration of the adversarial resource allocation problem.}
    \label{fig:illustration}
    \vspace{-0.2in}
\end{figure}

In this work, we consider centralized strategies. 
Namely, the \defender{} (resp. \attacker{}) decides the next allocation and sends instructions to the robots / resources to follow.
Consequently, the only intelligent agents are the \defender{} and the \attacker{}.
Our formulation also leads to feedback strategies that re-allocate resources based on the system state (the current allocation of the \attacker{} resources and the \defender{} resources). 
The re-allocation is done with all possible next actions of the opposing player in mind. 
This is a major difference from many prior works on resource allocation in the robotics community, where the focus has been either on achieving a desired terminal allocation that is fixed \cite{berman2009optimized,prorok2017impact}, or on scheduling to satisfy a time-varying but known demand (e.g., multiple traveling salesman problem)~\cite{khamis2015multi}.


\subsection{Related Work}

\paragraph*{Population model on graphs:}
The distributed resource allocation problem over a graph environment was proposed in~\cite{berman2009optimized}, where the authors developed stochastic control laws that drive the population of robots to a desired distribution to meet a \textit{static} demand.
The theory was later extended to accommodate heterogeneous robots and tasks with more diverse needs~\cite{prorok2017impact,ravichandar2020strata}.
However, the theoretical analysis in these works focused on the steady-state performance of the system, and a more delicate transient response to dynamically changing conditions was ignored. 
In contrast, our work focuses on the feedback mechanisms for a player to react to external inputs, but with the simplification of being centralized. Our work can be viewed as an ``outer loop'' that updates the desired allocation in response to adversarial actions, which the distributed control laws in \cite{berman2009optimized} can track as an ``inner loop'' at a faster time scale.

\paragraph*{Dynamic resource/task allocation:}
Dynamical aspect of the resource allocation problem has been studied in different ways.
Scheduling is one such formulation that considers tasks that must be completed in sequence \cite{tereshchuk2019efficient}.
On top of an efficient allocation algorithm, an adaptation mechanism is proposed in~\cite{tereshchuk2019efficient} which reacts to robot failures through a ``market-based" optimizer to re-allocate the leftover tasks.
A distributed resource allocation on a graph environment has also been studied with an adaptation mechanism \cite{lerman2006analysis}, where the population dynamics are controlled through the adaptation of individual behaviors based on local sensing.
These works provide scalable within-population interactions, but the adaptation schemes are purely reactive and do not contain any anticipation of the failure or changes that may occur in the future. 
In contrast, this paper emphasizes the between-population (\defender{} resources vs.~\attacker{} resources) strategic interactions, where each player selects its action based on the anticipated optimal reactions from the opposing player.%
\footnote{Note that in safety-critical systems, one can model the environment as an adversarial agent/team that seeks to undermine the performance of the deployed system.}

\paragraph*{Colonel Blotto Games:}
The static version of the adversarial resource allocation problem is commonly formulated as Colonel Blotto game~\cite{roberson2006colonel,powell2009sequential,chandan2020showing,kovenock2018optimal}. 
In the most standard version~\cite{gross1950continuous} of the game, two colonels allocate their resources to multiple locations. 
Whoever allocated more resource wins that location, and each colonel seeks to maximize the number of locations s/he wins. 
Many variants of the Colonel Blotto game have been studied, including asymmetric budget~\cite{roberson2006colonel}, asymmetric information~\cite{paarporn2019characterizing}, etc. 
However, most of the formulations in the existing literature consider static games, which assume that the desired allocation is achieved instantly and thus ignore the dynamics that are involved in the resource transportation. 
Although more recent works have considered dynamical extensions of Colonel Blotto games \cite{konrad2018budget,hajimirsaadeghi2017dynamic,klumpp2019dynamics}, their formulation does not capture the transportation of the resources in the environment.

\paragraph*{Preliminary work:}
The conference version of this work~\cite{shishika2022dynamic} introduced the \emph{dynamic Defender Attacker Blotto~(dDAB) game} that combines the ideas from Colonel Blotto games~\cite{gross1950continuous} and the population dynamics over graphs~\cite{prorok2017impact}. 
The conference version has identified the critical resource ratios (CRR) for a special class of graphs (ring graphs) and proposed a sampling-based algorithm that only provides certificates for the \attacker{}'s victory when the algorithm returns a solution.  
The analysis on the \defender{} side (e.g., necessary and sufficient conditions for the \defender's victory, the \defender's strategies, etc.) was not fully conducted in~\cite{shishika2022dynamic}. 
This paper provides a complete characterization of the dDAB game on any given graph.

\subsection{Contributions}

Our formulation yields feedback strategies that reallocate resources based on the evolving system state—namely, the current locations of the \attacker{} and the \defender{} resources. 
Unlike prior work that yields open-loop strategies against known demand~\cite{berman2009optimized,khamis2015multi}, 
we address the adversarial aspect of the proposed dDAB game by employing a novel reachability-based analysis.
Such approach leads to feedback strategies that reallocate resources with all possible next allocations of the opposing player in mind, thus providing worst-case guarantees.

To handle the game's temporal structure, we develop a set-based dynamic programming algorithm that recursively computes \emph{$k$-step safe sets}—\defender{} allocations that are \emph{necessary and sufficient} to maintain defense for $k$ time steps against any \attacker{} strategy. 
The proposed algorithm explicitly incorporates the traversability constraints of the resources, and exploits the geometric properties of the safe sets for computational efficiency. 
We further mitigate the curse of dimensionality by first analyzing no-splitting \attacker{} strategies, then generalizing to arbitrary strategies via a subteam superposition approach.

Our analysis leads to three key results:
\begin{enumerate}
    \item 
    Identification of the critical amount of \defender{} resources that is \emph{necessary and sufficient} for guaranteed defense over a given graph.
    \item 
    Synthesis of \emph{feedback strategies} that ensure successful defense against any \attacker{} strategy.
    \item 
    Formal proof that the \attacker{} gains no advantage by splitting its resources into subteams, along with \attacker{} strategies that guarantee its earliest victory when defense is infeasible.
\end{enumerate}
These results provide practical guidance for designing and deploying \defender{} robotic systems capable of provably securing a graph against intelligent \attacker{}s, with explicit guarantees on the required amount of resources.

\section{Problem Formulation}
\label{sec:formulation}
The \thegame{} is played between two players: the \defender{} and the \attacker{}.
The environment is represented as a directed graph $\graph = (\nodeset,\edgeset)$, where the $N$ nodes represent locations, and the directed edges represent the traversability among those locations.
We assume that $\graph$ is strongly connected~\cite{berman2009optimized}, i.e., every node is reachable from any other node.%
\footnote{
The assumption of strongly connected graph is used to avoid the degenerate cases with ``sinks" in the graph, which the \defender{} resource cannot get out from once reached. See~\figref{fig:sink_example} in Appendix~\ref{appdx:sink-example} for an example.
}
For notational simplicity, we assume that the two players share the same graph, but the present analysis easily extends to the case where the two players have different edge sets.

To capture the connectivity among the nodes, we define the graph adjacency matrix $A \in \reals^{N \times N}$ as follows:
\begin{equation*}
    \left[A\right]_{ij} = \left\{\begin{array}{cl}
         1 & \text{if $(j,i) \in \edgeset$},   \\
         0 & \text{otherwise}.
    \end{array} \right. 
\end{equation*}
The \textit{out-degree} of node $i$ is denoted as $\dout_i = \sum_j [A]_{ji}$, and its \textit{out-neighbors} is denoted as $\neighbor{i} = \{j\in \nodeset | (i,j) \in \edgeset\}$.

The total amount of resources for the \defender{} and the \attacker{} are denoted by $X\in \reals_{>0}$ and $Y\in \reals_{>0}$, respectively. 
For some time horizon $T$, the allocation of the \defender{}'s resources over the graph at time $t=0,1,\ldots, T$ is denoted by the state vector (allocation vector) $\x_t\in \reals^N$, which lies on a scaled simplex, such that $[\x_t]_i \geq 0$ and $\sum_i [\x_t]_i = X$.
The state vector (allocation vector) $\y_t\in \reals^N$ for the \attacker{} also satisfies the same conditions with $X$ replaced by $Y$.
{We use $\Delta_X$ and $\Delta_Y$ to denote the state space of the \defender{} and the \attacker{}.}
Note that continuous resources ($\x_t$ and $\y_t$ are continuous variables) are considered in this work.%
\footnote{Such an assumption on the state vector simplifies the analysis in \cite{berman2009optimized,prorok2017impact}, however, we will later show that our algorithms accommodate states that take discrete values.}

\subsection{Dynamics}
The major difference from the original Colonel Blotto game is that the \thegame{} is played over multiple time steps, and that the states evolve according to the following discrete-time dynamics:
\begin{equation}
    \x_{t+1} = K_t \x_t\;\;\;\text{and}\;\;\; \y_{t+1} = F_t \y_t,
    \label{eq:dynamics}
\end{equation}
where $K_t$ and $F_t$ represent the \emph{transition matrices} for the \defender{} and the \attacker{}, respectively. 
These matrices are left stochastic (column sum is unity), and their $ij$-th entry can take nonzero values only when $[A]_{ij}=1$. 
These matrices represent the action/control executed by the players.
For example, an action $K_t$ of the \defender{} is admissible if and only if it satisfies the following \emph{linear} constraints:
\begin{alignat}{2}
    &\quad K_t^\top \mf{1} = \mf{1}, && \label{eqn:stochastic-matrix-constraint}\\
    &\quad [K_t]_{ij} \geq 0, \qquad && \forall \; i, j \in \nodeset,
    \label{eqn:non-negativity-constraint}\\
    &\quad [K_t]_{ij} = 0, \qquad && \text{if } A_{ij} = 0.
    \label{eqn:adjacency-constraint}
\end{alignat}
The entry $[K_t]_{ij}$ denotes the fraction of resource on node $j$ to be transferred to node $i$ at the next time step.
We denote the admissible set for the matrices $K_t$ as $\mathcal{K}$, which
depends only on the underlying graph $\graph$ and is time-invariant.
The matrix $F_t$ for the \attacker{} also satisfies similar constraints, and
we denote the set of all admissible matrices $F_t$ as $\F$.%
\footnote{Under the assumption that the two players have the same graph, we have $\F = \K$. 
For consistency, we still use the notations of $\K$ and $\F$ to denote the two action spaces.}

\subsection{Terminal conditions \& Sequential Actions}
Similar to the Colonel Blotto game~\cite{gross1950continuous},
the engagement at each location is modeled solely based on the amount of resources. 
However, we evaluate the game outcome on a subset of nodes $\keyloc \subseteq \nodeset$, which we refer to as the key nodes%
\footnote{In a perimeter defense scenario, the key nodes can be the positions on the perimeter. 
When the \defender{} is defending a high-value asset, the key nodes can be the entrance points to the asset.}.
Specifically, the \defender{} successfully \emph{guards} a key node by allocating at least as much resource as the \attacker{} does, whereas the \attacker{} \emph{breaches} a key node by allocating more than what the \defender{} does.
For the \thegame{}, the \defender{} wants to prevent the \attacker{} from breaching any key node.
In this work, we mainly focus on a finite horizon~$T$.
The game terminates with the \attacker{}'s victory at the earliest time instance $t \in \{0, \ldots, T\}$ at which
\begin{equation}
    \exists \;i \in \keyloc, \st [\y_t]_i > [\x_t]_i.
    \label{eq:terminal_condition}
\end{equation}

The \defender{} wins the game if it can prevent the \attacker{} from achieving condition~\eqref{eq:terminal_condition} for all $t \!\in\! \{0, \ldots, T\}$. 
If the \defender{} can prevent~\eqref{eq:terminal_condition} for all time horizons $T \geq 1$, we say that the \defender{} can defend indefinitely.

The key node formulation provides a generalization to the prior work~\cite{shishika2022dynamic},
in which the \defender{} needs to defend \emph{all} nodes in the graph.

\subsection{Information Structure}
For the information structure, we assume that the players make decisions in sequence.
Specifically, the \defender{} acts first then the \attacker{} acts next, i.e., the \attacker{} selects its action after observing how the \defender{} allocated its resources.
The game outcome is evaluated after the \attacker's move.
To avoid the degenerate scenario where the \attacker{} wins immediately in the first time step, we let the \attacker{} specify its initial allocation $\yinit$, followed by the \defender{} freely picking its distribution $\x_0$ after observing $\yinit$.
The timeline of the dDAB game is presented in ~\figref{fig:time-line}.
In a realistic scenario where the two players make simultaneous actions,
our problem formulation corresponds to a worst-case scenario for the \defender{}.
Importantly, our setting accommodates state feedback strategies in contrast to previous results with constant action (transition) matrices \cite{berman2009optimized,prorok2017impact}.

\begin{figure}
    \centering
    \includegraphics[width=0.65\linewidth]{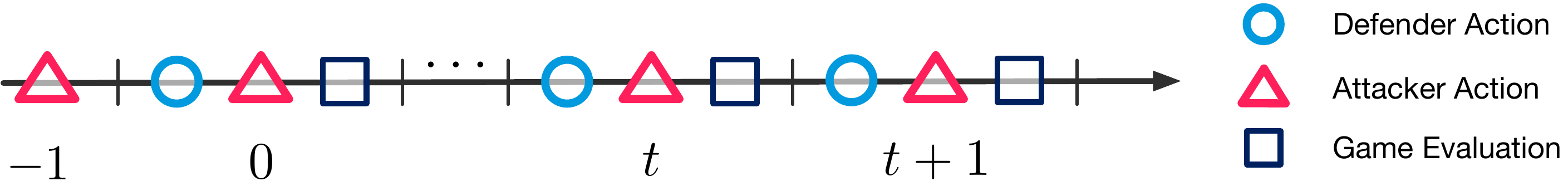}
    \caption{Sequence of events at every time step of the dDAB game. The \defender{} first moves its resources based on the observation of the current \attacker{} allocation. The \attacker{} then observes and reallocates. Finally, the game outcome at this time step is evaluated after the \attacker's move.}
    \label{fig:time-line}
\end{figure}

Finally, we consider centralized strategies in this work.
Specifically, the \defender{} (\attacker{}) serves as a coordinator, who decides the next allocation for its resources.
The allocation instructions, encoded as $K_t$ ($F_t$), are then sent to the resources (robots) to follow.

\subsection{A Simple Example}
\begin{figure*}[b]
    \centering
    \vspace{+5pt}
    \includegraphics[width=0.98\linewidth]{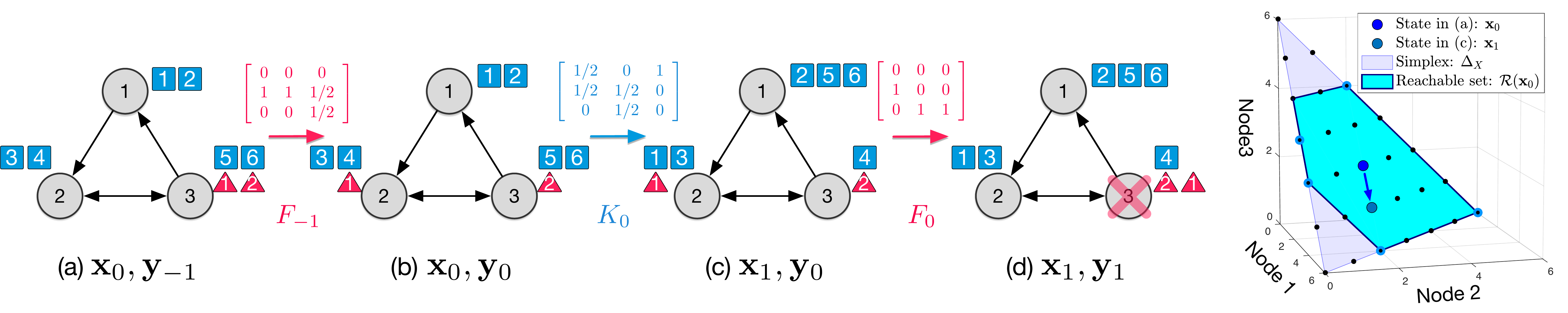}
    \vspace{-0.15in}
    \caption{An illustrative example of a three-node dDAB game, with \attacker{}'s victory at time $t=1$. Self-loop on each node is implied, and all nodes are key nodes. 
    The agents are indexed to illustrate their movements.
    Right-most plot presents the \defender's reachable set from the allocation in (a). The black dots indicate the discrete states if the \defender's resources consists of indivisible units/robots.
    }
    \label{fig:example}
    \vspace{-0.15in}
\end{figure*}

We present a three-node example in \figref{fig:example}, where all nodes are key nodes for simplicity, i.e., $\keyloc = \nodeset$.
In (a), the \attacker{} starts with an initial allocation of $\y_{-1} = [0,0,2]$, while the \defender{} selects $\x_0 = [2, 2, 2]$ as its starting configuration. 
In (b), after observing $\x_0$, the \attacker{} employs the red matrix $F_{-1}$ to update its allocation to $\y_0 = [0,1,1]$. 
In (c), the \defender{} redistributes its own resources via the blue matrix $K_0$.
The states depicted in (c) are $\x_1 = [3,2,1]$ and $\y_0=[0,1,1]$.
Finally, in (d), the \attacker{} observes that $\x_3$ has only one blue robot at node 3 and moves its Robot 1 from node 2 to node 3 to breach the node. 
Consequently, the game terminates at time $t=1$ with \attacker{}'s victory, concluding with the states $\x_1 = [3,2,1]$ and $\y_1 = [0,0,2]$.

\subsection{Research Problems}
Based on the discussion above, an instance of a \thegame{} is defined by: (i)~the available amount of resources $X$ and $Y$, (ii)~the graph $\graph$, and (iii) the required defense horizon.
Given a graph, our goal is to identify the necessary and sufficient amount of resources for the \defender{} to win the game, as well as its corresponding strategies.
To formalize the above goal, we introduce the following multiplicative factor.
\begin{definition}[Critical Resource Ratio]
\label{def:crr}
For a given graph $\graph$ and a time horizon $T$, the Critical Resource Ratio (\multfactor{}), $\factor{T}{} \geq 1$, is the \emph{smallest} positive number such that, if 
    $X \geq \factor{T}{} Y$,
then the \defender{} has a strategy to defend up to time step~$T$ against any admissible \attacker{} strategy that starts at any initial state $\yinit \in \Delta_Y$.
We use $\factor{\infty}{}$ to denote the \multfactor{} that enables the \defender{} to defend \emph{indefinitely}.
\end{definition}

Notice that the CRR defined above is the \emph{necessary and sufficient} amount of \defender{} resources to guarantee defense over the given time horizon for the given graph.

The two main questions we address in this work are:
\begin{problem}
    Given a graph and a finite horizon $T$, what is the \multfactor{} $\factor{T}{}$?
\end{problem}

\begin{problem}
    When $X\geq \factor{T}{}Y$, what is the \defender{} strategy that guarantees defense over $T$ time steps?
    Given insufficient amount of \defender{} resources, what is the optimal \attacker{} strategy to achieve the earliest possible breach?
\end{problem}

\section{Reachable Sets and Required Sets}
\label{sec:reachable-set}

In this section, we study the \defender{}'s allocation configurations that guarantee defense at the current time step and introduce key concepts essential for the subsequent analysis.
Most of the results are drawn from the conference version ~\cite{shishika2022dynamic}
and are included here for completeness.

\subsection{Reachable Sets}
\label{subsec:reachable_sets}
To better predict and understand how the allocation of resources evolves over time, we focus on the possible states that the \defender{} and \attacker{} can reach at the next step, i.e., their reachable sets.
Working with reachable sets offers two main advantages over working directly with the action spaces $\K$ and $\F$:
(i) the dimensionality of the reachable sets is significantly lower than that of the edge sets ($|\nodeset| \ll |\edgeset|$), and
(ii) the reachability analysis circumvents the non-uniqueness of actions that can produce a given transition from $\x_t$ to $\x_{t+1}$.
Since the dynamics of the two players are symmetric, we restrict our analysis to the \defender{}'s reachable sets and its action space $\mathcal{K}$.

\begin{definition}[Reachable Set from a Point]
The reachable set from a single point $\x_t$, denoted as $\RSet(\x_t)$, is the set of all states that the \defender{} can reach at the next time step with an admissible action. 
Formally,
\begin{equation}
    \RSet(\x_t) = \{ \x \;|\; \exists K\in\K \st \x =K \x_t \}.
\end{equation}
\end{definition}

\begin{remark}
    All points in the reachable set satisfy the conservation of resource. 
    That is, for all $\x_{t+1} \in \RSet(\x_t)$, we have that  $\mf{1}^\top \x_{t+1} = \mf{1}^\top \x_{t}$.
\end{remark}

To better understand the properties of the reachable sets, we first examine the structure of the action space.
Under the linear constraints in~\eqref{eqn:stochastic-matrix-constraint}--\eqref{eqn:adjacency-constraint}, the set of admissible actions $\mathcal{K}$ is a bounded polytope in the $|\edgeset|$-dimensional space.
We use the extreme points (vertices) of this polytope to characterize $\mathcal{K}$.

Given the admissible action space $\K$, we define the set of \textit{extreme actions} as
\begin{equation}\label{eqn:extreme-actions}
    \hat{\mathcal{K}} = \big\{K \in \mathcal{K} \;|\; [K]_{ij} \in \{0,1\} \big\}.
\end{equation}
In words, $\hatcalK$ contains all admissible actions $K$ whose entries are either 0 or 1. 
The cardinality of $\hatcalK$ is given by $\big\vert \hatcalK \big\vert = \prod_{j \in \nodeset} \dout_j$, where $\dout_j$ is the out-degree of node $j$.
We use $\ell$ to index the extreme actions in $\hatcalK$, 
i.e. $\hatcalK = \{\hatK^{(\ell)}\}_{\ell=1}^{|\hatcalK|}$.
The following theorem reveals the connection between the extreme actions and the admissible action set.

\begin{restatable}{theorem}{ea}
\label{thm:extreme-action}
The extreme actions defined in~\eqref{eqn:extreme-actions} are the vertices of the polytope $\mathcal{K}$. 
Formally,
\begin{equation}\label{eqn:convex-hull}
    \K = \convexhull{\hatcalK}.
\end{equation}
Consequently, for any admissible action $K \in \K$, 
there is a set of non-negative coefficients $\bm{\lambda}=\{\lambda^{(\ell)}\}_{\ell=1}^{|\hatcalK|}$ such that $\sum_{\ell=1}^{|\hatcalK|} \lambda^{(\ell)}=1$ and
\begin{equation}\label{eqn:action-convex-comb}
    K = \sum_{\ell=1}^{|\hatcalK|} \lambda^{(\ell)} \hatK^{(\ell)}.
\end{equation}
\end{restatable}
\begin{proof}
See Appendix~\ref{appdx:reachable-sets}.
\end{proof}

\begin{remark}
The extreme action set $\hatcalK$ depends only on the graph $\graph$, and it only needs to be constructed once.
\end{remark}


The extreme action set for the \attacker{} is denoted as $\hatcalF$ and is defined similarly;
we use $\{\hatF^{(r)}\}_{r=1}^{|\hatcalF|}$ to index the elements of $\hatcalF$.

\subsubsection{Reachable Sets as Polytopes}
The reachable set $\RSet(\x_t)$ is, in fact, a polytope in $\Delta_X$, and it can be viewed as a transformation performed on the action space $\K$.
Formally, we have the following lemma, which is a direct result of Theorem~\ref{thm:extreme-action}.

\begin{lemma}
    \label{lmm:RSet-Polytope}
    Given a point $\x_t$, the reachable set $\RSet(\x_t)$ is a polytope {given by $\RSet(\x_t) = \convexhull{\{\hatK^{(\ell)} \x_t\}_{\ell=1}^{|\hatcalK|}}$.} 
\end{lemma}
\begin{proof}
    For any $\x \in \RSet(\x_t)$, by definition, there is an action $K_t \in \K$, such that $\x = K_t \x_t$.
Based on the characterization of $\K$ in~\eqref{eqn:action-convex-comb}, this $\x$ can be represented as the following convex combination for some $\bm{\lambda}$:
\begin{equation}
    \label{eqn:characterization-reachable-set}
    \x = K \x_t =\bigg(\sum_{\ell=1}^{|\hatcalK|} \lambda^{(\ell)} \hatK^{(\ell)} \bigg)\x_t = \sum_{\ell=1}^{|\hatcalK|} \lambda^{(\ell)} \left(\hatK^{(\ell)} \x_t\right). 
\end{equation} 
Define $\vx^{(\ell)}_{t+1} = \hatK^{(\ell)} \x_t$ to be the state achieved by propagating $\x_t$ with the extreme action $\hatK^{(\ell)}$.
Then, the convex hull of these vertices gives us the polytope $\RSet(\x_t) = \convexhull{\{\vx_{t+1}^{(\ell)}\}_{\ell=1}^{|\hatcalK|}}$, which describes the set of states that the \defender{} at $\x_t$ can achieve at the next time step. 
\end{proof}

\figref{fig:example} presents an example of the reachable set for a three-node graph.
For discrete resources (robots) as illustrated in \figref{fig:example}(a), the \defender{} is able to achieve any discrete state (black dots) is are contained in the reachable set.


Using the same argument, we can compute the \attacker{} reachable set via $\RSet(\y_t) = \convexhull{\{\vy^{(r)}_{t+1}\}_r}$, where the vertices are given by $\vy^{(r)}_{t+1} = \hatF^{(r)} \y_t$ for $r=1,2,...,|\hatF|$.

Since any state in $\RSet(\x_t)$ can be reached at the next time step from $\x_t$, we view this polytope as the action space for   the \defender{} at state $\x_t$.
This definition of the action space resolves the two issues raised at the beginning of this section: dimensionality and nonuniqueness.

\subsubsection{Reachable Sets of Polytopes}
We extend the definition of the reachable set of a single point to the reachable set of a (potentially unbounded) set, which will play a significant role in our later analysis of the optimal strategies.
\begin{definition}[Reachable Set from a Set]
Given a set $P \subseteq \R^{n}_{\geq 0}$, the reachable set from this set, denoted as $\RSet(P)$, is the set of all states that the player can reach at the next time step with an admissible action starting from a state within $P$. 
Formally, 
\begin{equation}
    \RSet(P) =\{\x = K \x_t \;|\; K \in \K, ~ \x_t \in P\}.
\end{equation}
\end{definition}

\begin{lemma}
    \label{lmm:RSet-poly}
    Given a polytope $P$, the reachable set $\RSet(P)$ is also a polytope.
\end{lemma}
\begin{proof}
    Due to the resolution theorem~\cite{bertsimas1997introduction}, any point $\x_t \in P$ can be expressed as 
    \begin{equation*}
        \x_t = \sum_{r=1}^R \theta^{[r]} \x^{[r]} + \sum_{m=1}^M \phi^{[m]} \h^{[m]},
    \end{equation*}
    where $\{\x^{[r]}\}_r$ is the set of vertices of $P$ and $\{\h^{[m]}\}_m$ is the set of extreme rays.
    Then, it is straightforward to show that 
    \begin{equation*}
        \RSet(P) = \convexhull{\{\hat{K}^{[\ell]} \x^{[r]}\}_{\ell, r}} + \conichull{\{\hat{K}^{[\ell]} \h^{[m]}\}_{\ell, m}},
    \end{equation*}
    where $\mathrm{Cone}$ represents the conic hull of the rays and the summation is a Minkowski sum.
\end{proof}

\subsection{Required Set}
In this section, we consider the \defender{}'s selection of $\x_{t+1}$ after observing the \attacker{}'s current allocation $\y_t$. 
The goal is to identify the set of the \defender{}'s feasible states $\x_{t+1}$ such that the \attacker{}, upon observing $\x_{t+1}$, cannot select an action $\y_{t+1} \in \RSet(\y_t)$ that results in a successful breach of any key node.

For the \defender{} to defend every key node at time $t+1$, it is necessary and sufficient that the allocation vector $\x_{t+1}$ matches or outnumbers $\y_{t+1}$ at every key node $i \in \keyloc$: 
\begin{equation}
    [\x_{t+1}]_i\geq[\y_{t+1}]_i \quad \forall i \in \keyloc.
    \label{eq:vert_condition}
\end{equation}

Since the \attacker{} takes its action after observing the \defender{}'s allocation $\x_{t+1}$, the question is whether the \defender{} can select an allocation $\x_{t+1}$ such that \eqref{eq:vert_condition} is true for all $\y_{t+1} \in \RSet(\y_t)$.
This observation leads to the following condition for selecting $\x_{t+1}$ to guarantee defense at time $t+1$:
\begin{equation}
    \label{eqn:xreq_opt_R-1}
    [\x_{t+1}]_i \geq \max_{\y_{t+1} \in \RSet(\y_t)} [\y_{t+1}]_i \quad  \forall i \in \keyloc.
\end{equation}

Since $ \RSet(\y_t)$ is a bounded polytope (Lemma~\ref{lmm:RSet-Polytope}), the optimization  $\max_{\y_{t+1} \in \RSet(\y_t)} [\y_{t+1}]_i$ can be viewed as a linear program, 
whose optimum is attained at one of the vertices of $\RSet(\y_t)$. 
Consequently, we define the minimum required \defender{} resources at $t+1$ as $\xreq_{t+1}$, whose elements are
\begin{equation}\label{eqn:x_req}
    [\xreq_{t+1}]_i = \left \{ 
    \begin{array}{ll}
       \max_r \left[ \vy_{t+1}^{(r)} \right]_i \qquad  &  \text{if }i \in \keyloc\\
        0 & \text{otherwise}
    \end{array}
    \right .
    ,
\end{equation}
where $\big\{\vy^{(r)}_{t+1}\big\}_r = \big\{\hatF^{(r)} \y_t\big\}_r$ are the vertices of $\RSet(\y_t)$.
Then, the condition in~\eqref{eqn:xreq_opt_R-1} can be expressed in the following (component-wise) vector inequality form 
\begin{equation}
    \label{eqn:xreq_opt_R}
    \x_{t+1} \geq \xreq_{t+1}.
\end{equation}

\begin{remark}
The \defender's minimum required resource at the \emph{next} time step, $\xreq_{t+1}=\xreq_{t+1}(\y_t)$, is a function of the \attacker{}'s \emph{current} state, $\y_t$.
\end{remark}

We now claim that the \defender{} can guarantee defense at $t\!+\!1$ by selecting $\x_{t+1}$ inside the polytope $\polyreq(\y_t)$, which is defined as follows.

\begin{definition}[Required Set]
    Given the \attacker{}'s allocation $\y_t$ at time $t$,  the required set for the \defender{} at time $t+1$ is defined as:
    \begin{equation}\label{eqn:poly_req}
        \polyreq(\y_t) \triangleq \{ \x_{t+1} \;|\; [\x_{t+1}]_i\geq [\xreq_{t+1}(\y_t)]_i,\; \forall\; i\in \nodeset \}.
    \end{equation}
\end{definition}

\begin{proposition}
    \label{prop:p-req}
    The condition $\x_{t+1} \in \polyreq(\y_{t})$ is necessary and sufficient for the \defender{} to defend time step $t+1$. 
\end{proposition}

\begin{proof}
From~\eqref{eqn:xreq_opt_R-1} and the definition of the \attacker{}’s reachable set, the \defender{} can guarantee that key node $i$ is defended against all feasible \attacker{} allocations at $t\!+\!1$ if it allocates at least $[\xreq_{t+1}]_i$ to that node. This establishes \emph{sufficiency}.

Suppose the \defender{} allocates $[\x_{t+1}]i < [\xreq_{t+1}]_i$. Then, the condition~\eqref{eqn:xreq_opt_R-1} is violated, and there exists a vertex $\vy_{t+1}^{(r)}$ of the \attacker{}’s reachable set such that $[\vy_{t+1}^{(r)}]_i > [\x_{t+1}]_i$. 
In this case, after observing the \defender{}'s allocation, the \attacker{} can select a feasible action (e.g., $\hatF^{(r)}$) to reach $\vy_{t+1}^{(r)} \in \RSet(\y_t)$ and breach key node $i$. This establishes \emph{necessity}.
\end{proof}

\begin{remark}
    \label{rmk:preq-equivalent}
    The required set $\polyreq(\y_t)$ can be equivalently expressed as
    \begin{align*}
        \polyreq(\y_t) &= \{\x_{t+1} \big \vert [\x_{t+1}]_i \geq \max_{\y_{t+1} \in \RSet(\y_t)} [\y_{t+1}]_i, \; \forall i \in \keyloc\}\\
        &= \{\x_{t+1} \big\vert [\x_{t+1}]_i \geq \max_{F_{t}\in \F} \; [F_t \y_{t}]_i, \; \forall i \in \keyloc\}.
    \end{align*}
\end{remark}

In other words, as long as the \defender{} can reach an allocation within the required set $\polyreq(\y_t)$, it is guaranteed to be safe at the next time step $t\!+\!1$. Conversely, if the \defender{} fails to achieve such an allocation, it will lose the game at $t\!+\!1$ against a rational \attacker{}.

There are three possible reasons the \defender{} may fail to reach the required set:

\begin{enumerate}
    \item  \textit{Insufficient resources}: 
    The total required resource, $\xmagreq_{t+1} = \mathbf{1}^\top \xreq_{t+1}$, depends on the graph $\graph$ and the current \attacker{} allocation $\y_t$.  
    If $\xmagreq_{t+1} > X$, then no \defender{} strategy can guarantee defense, regardless of the current allocation $\x_t$.
    \item  \textit{Suboptimal strategy}: An allocation within the required set is feasible from the current \defender{} state $\x_t$, but the \defender{} selects a bad next allocation outside of $\polyreq(\y_t)$.
    \item Bad current allocation:
    The \defender{} has sufficient total resource ($X \ge \xmagreq_{t+1}$), but its current state $\x_t$ does not allow it to reach any point in the required set. That is, $\RSet(\x_t) \cap \polyreq(\y_t) = \emptyset$.
\end{enumerate}

\begin{figure}[b]
    \centering
    \includegraphics[width = 0.35\textwidth]{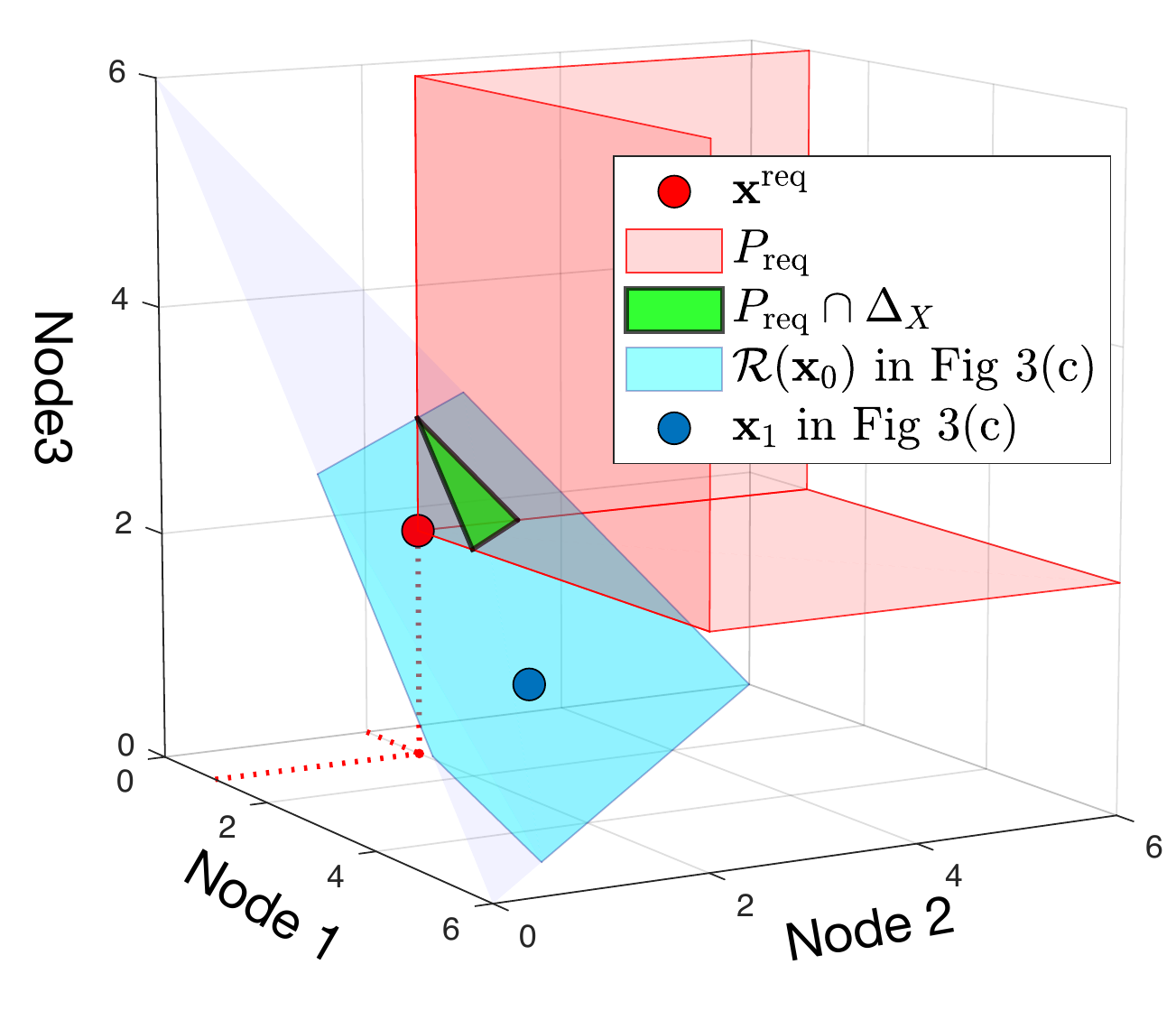}
    \vspace{-0.1in}
    \caption{Illustration of the required set $\polyreq$ for the graph in~\figref{fig:example} with $\y_1 = [0,1,1]^\top$ and $X=6$.
    We conclude that $\xreq_1=[1,2,2]^\top$ and hence $\xmagreq=4$. The red surface shows the boundary of $\polyreq$.}
    \label{fig:required_set}
\end{figure}

\subsection{Example}
For the \attacker{} configuration $\y_0$ in \figref{fig:example}(c), the required \defender{} allocation at the next time step, denoted $\xreq_1$, can be determined as follows.
At \textbf{node~1}, the \attacker{} can place at most one robot by moving Robot 2 from node 3.
At \textbf{node~2}, the \attacker{} can place two robots by keeping Robot 1 at node 2 and moving Robot 2 there.
Similarly, at \textbf{node 3}, two robots can be placed.
Thus, we can conclude that $\xreq_1 = [1, 2 , 2]$.
One can observe that the \defender{} allocation $\x_1=[3, 2, 1]$ in (c) falls short for $\xreq_1$ at node 3, and thus is breached by the \attacker{} in (d).

\figref{fig:required_set} highlights in green the intersection between the required set $\polyreq(\y_0)$ and the \defender{}'s state space $\Delta_X$.
For the action sequence presented in \figref{fig:example}, one can observe that the selected \defender{} state $\x_1$ lies outside the green intersection, leading to the \defender{}'s defeat at time $t=1$. 
This failure corresponds to case (2) suboptimal strategy, since the reachable set (cyan) does intersect the required set, but the \defender{} chooses a suboptimal allocation.

\subsection{Degenerate Parameter Regime}
Notice that $\xmagreq_{t+1}=\mf{1}^\top \xreq_{t+1}$ depends on $\graph$ and $\y_{t}$.
Clearly, the \defender{} does not have a strategy to guarantee defense if $\xmagreq_{t+1}>X$.
This immediately leads to the following result.

\begin{proposition}[Degenerate Parameter Regime \cite{shishika2022dynamic}]
\label{thm:degenerate_regime}
Let $\doutkey{i}$ denote the number of key nodes connected to node $i$, i.e.,
$\doutkey{i} \triangleq |\{j \in \keyloc \mid (i,j) \in \edgeset\}|$.
Define $\doutmax \triangleq \max_{i \in \nodeset} \doutkey{i}$
as the maximum number of key nodes adjacent to any node.
If the total resources satisfy
\begin{equation}
    X < \doutmax Y,
\end{equation}
then the \attacker{} can win the game at time step $t=0$.
\end{proposition}

\begin{proof}
Let $i^\star \in \nodeset$ be a node achieving $\doutkey{i^\star} = \doutmax$.
Consider the strategy where the \attacker{} initializes the game with $\yinit$ that concentrates all its resources at node $i^\star$.
At time $t=0$, the \defender{} allocates resources $\x_0$ to the key nodes adjacent to $i^\star$.
There are $\doutmax$ such key nodes.
To prevent an immediate loss, the \defender{} must allocate at least $Y$ units of resources to each of these nodes, which requires at least $\doutmax Y$ units of resources.

If $X < \doutmax Y$, then there exists at least one neighboring key node $j$ such that $x_0^j < Y$.
After observing $\x_0$, the \attacker{} moves all resources from $i^\star$ to node $j$ and captures it immediately.
Hence the \attacker{} can win at time $t=0$.
\end{proof}

Based on Proposition~\ref{thm:degenerate_regime}, the rest of the paper focuses on the case where 
\begin{equation*}
    X \geq \doutmax Y.
\end{equation*}

\section{No-Splitting Attacker}\label{sec:nsp}
This section develops the tools to construct optimal feedback strategies for the \defender{} and the \attacker. 
We first focus on the case where the \attacker{} resources move as a single concentrated group (a blob).
In~\secref{sec:split}, we generalize the results to scenarios where the \attacker{} splits its resource into multiple subgroups.
Note that throughout this paper, we do not restrict the \defender{}'s allocation strategies.

Let $\e_i \in \mathbb{R}^n$ be the unit vector with its $i$-th element equal to one.
In the sequel, we use the shorthand ${\y^{(i)} = Y\e_i}$ to denote the \attacker{} allocation that is fully concentrated on node $i$.

\begin{definition}[No-Splitting Attacker Strategy] 
A no-splitting \attacker{} strategy selects its action exclusively from the set of extreme actions, i.e., $F_t \in \hatcalF$ for all $t$. 
\end{definition}

\noindent Under a no-splitting \attacker{} strategy, if the initial \attacker{} allocation is fully concentrated, then the \attacker{}'s state remains concentrated for all time steps, i.e., $\y_t = \y^{(i_t)} \triangleq Y\e_{i_t}$, where $i_t$ denotes the location of the \attacker{}'s concentrated resources.

\subsection{K-step Safe Sets}
The key challenge we address in this section is the fact that selecting a state in the required set does not imply that the \defender{} can do so again in the next time step.\footnote{See Fig.~4 of \cite{shishika2022dynamic} for an example of a situation where single-step defense can be achieved, but the \attacker{} is able to breach after two steps.} 
As an example, for the \defender{} to guarantee defense over the next \emph{two} time steps starting from the current allocations $\x_t$ and $\y_t$, the following condition is necessary:
\begin{subequations}
    \begin{eqnarray}
        &\exists\; \x_{t+1}\in\polyreq(\y_t)\cap\RSet(\x_t)\;\text{s.t.}\;
        \label{eqn:2-step-a}
        \\
        &\polyreq(\y_{t+1})\cap\RSet(\x_{t+1})\neq\emptyset,\; \forall\; \y_{t+1}\in\RSet(\y_t).
        \label{eqn:2-step-b}
    \end{eqnarray}
\end{subequations}
In words, \eqref{eqn:2-step-a} ensures that the \defender{} selects a reachable allocation $\x_{t+1}$ that guarantees defense at $t+1$, 
while \eqref{eqn:2-step-b} ensures that the selected $\x_{t+1}$ allows transition to a state $\x_{t+2}$ that guarantees defense at $t+2$ against all potential \attacker{} allocation $\y_{t+1} \in \RSet(\y_t)$.

The need to account for all possible future \attacker{} actions quickly renders this formulation intractable beyond two steps.
To overcome this, we introduce the notion of $k$-step safe sets, later formalized as Q-sets.

\begin{definition}[$k$-step Safe Set]
\label{def:Q_set}
Let $\y_{t-1} = \y^{(i)}$ be the \attacker{} state concentrated at node $i$. The set $\nsp{k}{i}  \subseteq \mathbb{R}_{\geq 0}^{|\nodeset|}$ is defined such that $\x_t \in \nsp{k}{i}$ if and only if there exists a \defender{} strategy that can defend against any no-splitting \attacker{} strategy through time step $t+k$ (inclusive).%
\footnote{There may exist a strategy for the \attacker{} that breaches the system at time $t+k+1$, but not before.} 
\end{definition}

While the above definition introduces the concept of multi-step safe sets, we next present a recursive formulation for constructing Q-sets in the following theorem.

\begin{theorem}
    \label{thm:recursive-form}
    The following recursive expression provides the $k$-step safe set:
    \begin{subequations}
        \label{eqn:Q-set-def}
        \begin{alignat}{2}
            \nsp{0}{i} &= \polyreq(\y^{(i)}),
            \label{eqn:Q-set-def-0}\\
            \nsp{k}{i} &= \Big\{\x \big\vert  \x \in \polyreq(\y^{(i)}) \wedge \RSet(\x) \cap \nsp{k-1}{\ip} \ne \emptyset ~\forall \ip \in \neighbor{i} \Big\},   \qquad \forall ~ k \geq 1,  \label{eqn:Q-set-def-k}
        \end{alignat}
    \end{subequations}
    {where $\neighbor{i}$ is the set of out-neighbors of node $i$.} 
\end{theorem}
\begin{proof}
    The proof proceeds by induction. 
    The base case in~\eqref{eqn:Q-set-def-0} matches the required condition for single-step safety (cf. Proposition~\ref{prop:p-req}).
    For the inductive step,~\eqref{eqn:Q-set-def-k} requires that:
    $\x$ defends against an \attacker{} concentrated at node $i$ at the current step, i.e., $\x \in \polyreq(\y^{(i)})$; 
    meanwhile, 
    for every possible next \attacker{} node $j \in \neighbor{i}$, the \defender{} can reach the corresponding $(k\!-\!1)$-step safe set, i.e., $\RSet(\x) \cap \nsp{k-1}{j} \ne \emptyset$.
    Formal proofs of sufficiency (Lemma~\ref{lmm:Q-suf}) and necessity (Lemma~\ref{lmm:Q-nes}) are provided below.
\end{proof}

\begin{lemma}[Sufficiency of Q-sets]
    \label{lmm:Q-suf}
    Let the Q-sets be defined in~\eqref{eqn:Q-set-def}, and suppose that the \attacker{} starts with $\y_{t-1} = \y^{(i)}$. 
    Then, by having $\x_t\in \nsp{k}{i}$, the \defender{} can defend at least until time step $t+k$. 
\end{lemma}

\begin{proof}
We provide a proof by induction. 

\textit{Base Case:} When $k = 0$, we have $\x_t \in \nsp{0}{i} = \polyreq(\y^{(i)})$.
From Proposition~\ref{prop:p-req}, the defense is guaranteed at time $t$.

\textit{Inductive hypothesis:} 
Suppose that for some $k\geq 1$, and for all $i \in \nodeset$, the condition $\x_t \in \nsp{k}{i}$ guarantees defense until time $t+k$ given that $\y_{t-1} = \y^{(i)}$.

\textit{Induction:}
Given $\y_{t-1} = \y^{(i)}$,
we let $\x_t \in \nsp{k+1}{i}$. 
Under the no-splitting strategy, suppose that the \attacker{} selects $\y_t = \y^{(\ip)}$, for some arbitrary $\ip \in \neighbor{i}$.
The \attacker{} cannot immediately win with this (or any other) action since the \defender{} state $\x_t \in \nsp{k+1}{i} \subseteq \polyreq(\y_{t-1})$ guarantees defense at time step $t$.
After observing $\y_t = \y^{(\ip)}$, we let the \defender{} select its next state so that $\x_{t+1} \!\in\! \RSet(\x_t) \cap \nsp{k}{\ip}$.
This new selection is reachable since $\x_t \in \nsp{k+1}{i}$ ensures that $\RSet(\x_t) \cap \nsp{k}{\ip} \ne \emptyset$ (from~\eqref{eqn:Q-set-def-k}).
After the \defender{}'s action, we are at a situation where $\y_t = \y^{(\ip)}$ and $\x_{t+1} \in \nsp{k}{\ip}$.
From the inductive hypothesis, the \defender{} can defend another $k$ steps from this time on. 
The \defender{} can thus defend until time step $t\!+\!k\!+\!1$.

\end{proof}

\begin{lemma}[Necessity of Q-sets]
    \label{lmm:Q-nes}
    Let the Q-sets be defined in~\eqref{eqn:Q-set-def}, and suppose that the \attacker{} starts with $\y_{t-1} = \y^{(i)}$. 
    If $\x_t\notin \nsp{k}{i}$, the \attacker{} can win the game before or at time step $t+k$. 
\end{lemma}

\begin{proof}
We prove this lemma via an inductive argument.

\textit{Base case:} Suppose $\x_t \notin \nsp{0}{i} = \polyreq(\y^{(i)})$. 
Then, by the construction of $\polyreq(\y^{(i)})$, there exists $\ip \in \neighbor{i}$ such that $\y_t = \y^{(\ip)}$ defeats $\x_t$ on node $\ip$.%
\footnote{
If $\x_t \notin \polyreq(\y^{(i)})$, we know that there exists at least one $\y_t \in \RSet(\y^{(i)})$ that breaches $\x_t$, and this $\y_t$ is not necessarily a concentrated configuration. 
Suppose this (potentially split) $\y_t$ defeats $\x_t$ on node $\ip$. 
Since we are starting from a concentrated state $\y_{t-1} = \y^{(i)}$, the \attacker{} can move all its resource to the same node $\ip$, and this concentrated state would also breach node $\ip$.
}
This corresponds to a \defender{} defeat at time $t$.

\textit{Inductive hypothesis:} Suppose that, for all $i$, $\x_t \notin \nsp{k}{i}$ implies that the \attacker{} with state $\y_{t-1} = \y^{(i)}$ can win the game before or at time step $t+k$.

\textit{Induction:}
Let the \attacker{} start with $\y_{t-1} = \y^{(i)}$ and the \defender{} select $\x_t \notin \nsp{k+1}{i}$. 
From the definition of $\nsp{k+1}{i}$, we have either of the following two cases: 
(i) $\x_t \notin \polyreq(\y^{(i)})$, which leads to an immediate defeat at $t$; or
(ii) there exists $\ip \in \neighbor{i}$, such that $\RSet(\x_t) \cap \nsp{k}{\ip} = \emptyset$.
In the latter case, the \attacker{} can move to $\y_t = \y^{(\ip)}$. 
Then, for all possible next \defender{} allocation $\x_{t+1} \in \RSet(\x_0)$, we have that $\x_{t+1} \notin \nsp{k}{\ip}$.
From the inductive hypothesis, the \defender{} will be defeated within 
$k$ steps from this time $t+1$. Thus, the \attacker{} can win the game before or at time step $t+k+1$.

\end{proof}

Next, we present two important properties of the Q-sets.
\begin{remark}\label{rem:monotonicity}
For a fixed node $i \in \nodeset$, the sequence $\big\{\nsp{k}{i}\big\}_{k}$ is a decreasing sequence of sets. Formally, for all $i\in \nodeset$ and $k \geq 0$,
    \begin{equation}
        \label{eqn:Q-set-monotonicity}
        \nsp{k+1}{i} \subseteq \nsp{k}{i}.
    \end{equation}
\end{remark}
The above remark follows directly from the definition of the Q-sets. 
That is, if the \defender{} can defend $k+1$ steps from some state, then it can clearly defend $k$ steps.

\begin{restatable}{theorem}{py}
\label{thm:Q-polytope}
All Q-sets are polytopes. 
\end{restatable}
We delay the proof of Theorem~\ref{thm:Q-polytope} to the algorithmic section, where we introduce additional tools to characterize and efficiently construct the Q-sets.

\subsection{Indefinite Defense}
The recursive definition of the Q-sets in~\eqref{eqn:Q-set-def} can be viewed as an operator mapping from $\big(2^{\Delta_X} \big)^{|\nodeset|}$ to itself, where $2^S$ denotes the power set of set~$S$.
Consequently, \eqref{eqn:Q-set-def} can be viewed as an iterative algorithm, and its fixed point(s) is therefore of great interest to study.
Note that a fixed point of  \eqref{eqn:Q-set-def} is an element in $\big(2^{\Delta_X} \big)^{|\nodeset|}$.

\begin{definition}[Indefinite Safe Set]
    \label{def:Q_inf}
    We define the indefinite safe sets $\nsp{\infty}{i} \subseteq \Delta_X$ for $i\in\nodeset$ as follows:
    \begin{equation}
        \label{eqn:q-inf-def}
        \nsp{\infty}{i} = \bigcap_{k \geq 0} \nsp{k}{i}.
    \end{equation}
\end{definition}

\begin{remark}
    Since the Q-sets are nested (descending), the above definition is equivalent to $\nsp{\infty}{i} = \lim_{k \to \infty} \nsp{k}{i}$.
\end{remark}

\begin{remark}
    The indefinite safe sets are either all empty or all nonempty. In the first case, the \defender{} cannot defend indefinitely with a finite amount of resource.   
\end{remark}

The first natural question is whether the collection of indefinite safe sets defined in~\eqref{eqn:q-inf-def} is a fixed point of the recursive formula in~\eqref{eqn:Q-set-def}.
\begin{restatable}{theorem}{fp}
    \label{thm:fixed-point}
   If the indefinite safe sets defined in~\eqref{eqn:q-inf-def} are nonempty, they satisfy the following fixed point relation for all nodes $i \in \nodeset$:
    \begin{equation}
        \label{eqn:q-inf-fixed-point}
        \nsp{\infty}{i} \!= \!\Big\{\x \big\vert \x \in \polyreq(\y^{(i)}) \and \RSet(\x) \cap \nsp{\infty}{\ip} \ne \emptyset ~\forall \ip \in \neighbor{i} \Big\}.
    \end{equation}
\end{restatable}
\begin{proof}
    See Appendix~\ref{appdx-sec:fixed-point}.
\end{proof}

In the following theorem, we formalize the natural conjecture that indefinite safe sets  guarantee an indefinite defense for the \defender{}.

\begin{theorem}
    \label{thm:q_inf}
    If $\nsp{\infty}{i}\neq\emptyset$,
    then $\x_t \in \nsp{\infty}{i}$ is necessary and sufficient for indefinite defense given that the \attacker{} is at $\y_{t-1} = \y^{(i)}$.
\end{theorem}

\begin{proof}
    The necessity is straightforward. If $\x_t \notin \nsp{\infty}{i}$, then $\x_t \notin \nsp{k}{i}$ for some finite $k$. 
    From the necessity of the $k$-step safe sets, we know that the \defender{} will be defeated within $k$ steps.

    For the sufficiency, suppose at time step $t$, the system is at the state $\y_{t-1} = \y^{(i)}$ and $\x_t \in \nsp{\infty}{i}$.
    Since $\nsp{\infty}{i} \subseteq \polyreq(\y^{(i)})$, the \defender{} can defend at least the current time step $t$.
    Next, suppose the \attacker{} moves to $\y_{t} = \y^{(\ip)}$, where $\ip \in \neighbor{i}$.
    From~\eqref{eqn:q-inf-fixed-point}, there is a state $\x_{t+1} \in \nsp{\infty}{\ip}$ that is reachable from $\x_t$.
    Since $\nsp{\infty}{\ip} \subseteq \polyreq(\y^{(\ip)})$, the defender can also defend the time step $t+1$. 
    Through mathematical induction, one can easily argue that being in $\nsp{\infty}{i}$ when $\y_{t-1} = \y^{(i)}$ guarantees indefinite defense for all \emph{no-splitting} \attacker{} strategies.
    
\end{proof}

The conditions on the graph that guarantee convergence of the iterative algorithm in~\eqref{eqn:Q-set-def} as well as the conditions for the existence of such fixed point(s) is an ongoing research. 
Note that not all graphs have such a fixed point, for example, a sink graph (see Figure~\ref{fig:sink_example} in Appendix~\ref{appdx:sink-example}) does not have one, since it requires infinite \defender{} resources to guard indefinitely. 
Empirically, we found that for all strongly-connected and undirected graphs, the iterative algorithm in~\eqref{eqn:Q-set-def} converges within $N$ iterations, where $N=|\nodeset|$ is the number of nodes.
A follow-up work would be to establish the convergence guarantees.

\subsection{Q-Set Propagation}
The Q-set propagation process is described in \algref{alg:Q-prop}, which takes three inputs: the graph environment $\mathcal{G}$, the \attacker{} total resource  $Y$, and the horizon of the game~$T$.
We assume that the players do not consider their performance beyond~$T$, and therefore, Q-sets are only computed up to this horizon.
The algorithm applies the recursion in~\eqref{eqn:Q-set-def} to construct the Q-sets for each node. 
In practice, a numerically efficient implementation uses an equivalent but computationally friendly formulation~\eqref{eqn:q-set-comp} in Section~\ref{sec:algorithm}.

The iterative construction terminates if the Q-sets converge, as checked in line~4 of the algorithm\footnote{
Since all Q-sets are polytopes, one can simply check the vertices (and extreme rays) for convergence.
}. 
In this case, we can conclude that the \defender{} has a strategy to defend \textit{indefinitely} against all no-splitting \attacker{} strategies. 
The output $k_\infty$ gives the smallest finite number such that $\nsp{k_\infty}{i}=\nsp{\infty}{i}$ for all $i\in\nodeset$.
For the remainder of the paper, when \algref{alg:Q-prop} converges, we refer to the converged Q-sets as $\{\nsp{\infty}{i}\}_i$ for notational simplicity.

\begin{algorithm}[ht]
\caption{Q-Prop}
\label{alg:Q-prop}
\SetAlgoLined
\SetKwInput{KwInputs}{Inputs}
\KwInputs{Graph $\graph$, \attacker{} total resource $Y$, game horizon $\hmax$\;}
Set $k_\infty = \infty$ and set $\nsp{0}{i} = \polyreq(\y^{(i)})$ ~for all $i \in \nodeset$\;
\For{$k = 1$ to $\hmax$}{
    Construct $\nsp{k}{i}$ using~\eqref{eqn:q-set-comp} for all $i \in \nodeset$\;
    \If{$\nsp{k}{i} = \nsp{k-1}{i}$ for all $i \in \nodeset$}{
    $k_{\infty} = k-1$\;
    \textbf{Break}\;
    }
}
\textbf{Return:} $\{\nsp{k}{i}\}_{i,k}$,  $k_{\infty}$
\end{algorithm}



\subsection{K-step Strategies}
The proof of Theorem~\ref{thm:recursive-form} provides a guideline for the strategies that the \defender{} and the \attacker{} would deploy under the no-splitting assumption. 
We first summarize the \defender{} strategy in the following two algorithms.

\begin{algorithm}[htb]
\caption{Initial Defender Allocation (against No-Splitting Attacker)}
\label{alg:defender-init-nsp}
\SetAlgoLined
\SetKwInput{KwInputs}{Inputs}
\KwInputs{Graph $\graph$, 
\defender{} total resource $X$,
\attacker{} total resource $Y$,
\attacker{} initial allocation $\yinit = \y^{(i_{-1})}$,
game horizon $\hmax$\;}

Construct Q-sets via~\algref{alg:Q-prop}\;
$ k_{\mathrm{max}, 0} \gets \argmax_k \left\{k \leq \min\{\hmax, k_\infty\} | \Delta_X \cap \nsp{k}{i_{-1}} \ne \emptyset \right\}$ 
\Comment{{\footnotesize\texttt{find the longest defense time}}}

$\x_{0} \gets$ any element in $\nsp{k_{\max,0}}{i_{-1}}$

\textbf{Return:} Initial allocation $\x_{0}$, guaranteed defense time $k_{\max,0}$
\end{algorithm}
\begin{algorithm}[htb]
\caption{Feedback Defender Strategy (against No-Splitting Attacker)}
\label{alg:defender-strategy-nsp}
\SetAlgoLined
\SetKwInput{KwInputs}{Inputs}
\KwInputs{Q-sets,
previous \defender{} allocation $\x_{t-1}$, 
previous \attacker{} allocation $\y_{t-1} = \y^{(i_{t-1})}$,
game horizon $\hmax$\;}

$k_{\max, t} \gets \argmax_k \left \{k \leq \min\{\hmax, k_\infty\} | \RSet(\x_{t-1}) \cap \nsp{k}{i_{t-1}} \ne \emptyset \right\}$ \Comment{{\footnotesize\texttt{exploit \attacker{}'s mistake}}} 
\\
$\x_{t} \gets$ any element in $\nsp{k_{\max,t}}{i_{t-1}}$\;

\textbf{Return:} New allocation $\x_{t}$, guaranteed defense time $k_{\max,t}$
\end{algorithm}

\algref{alg:defender-strategy-nsp} presents the feedback strategy for the \defender{}.  
Suppose $k_{\max, 0} = k_\infty$ in \algref{alg:defender-init-nsp}, then the \defender{} can indefinitely defend regardless of the \attacker{}'s no-splitting strategy. 
In this case, the \defender{} observes $\y_{t-1} = \y^{(i_{t-1})}$ and reallocates its resources to the corresponding Q-set: $\nsp{k_\infty}{i_{t-1}} = \nsp{\infty}{i_{t-1}}$. 

On the other hand, if Algorithm~\ref{alg:defender-init-nsp} outputs $k_{\max, 0} < k_\infty$, then either \algref{alg:Q-prop} did not converge, or the \defender{} does not have enough resource to achieve indefinite defense.
By the construction of Q-sets, the \defender{} has a guarantee to defend up to time step $t=k_{\max, 0}$.
If we also have $k_{\max, 0} <T$, then the \attacker{} will identify a strategy to win at $t=k_{\max, 0}+1$ (shown later in Algorithms~\ref{alg:attacker-init} and~\ref{alg:attacker-strategy}).
Under the rational strategies by both players, $k_{\max, t}$ will reduce by 1 at each time step, and the game terminates with \attacker{}'s win at $t=k_{\max, 0}+1 \leq T$.
However, if the \attacker{} does not play rationally, the \defender{} may be able to delay the breaching. 
The search / optimization performed in line~1 of Algorithm~\ref{alg:defender-strategy-nsp} ensures that the \defender{} exploits such opportunity.\footnote{Note that if $k_{\max, 0}=T<k_\infty$, we do not have an estimate of when the \attacker{} will be able to breach, even if the game continued beyond $t=T$. However, the \defender{} still has a guarantee to defend up to time step $T$, and that is sufficient to identify the outcome of the finite-horizon game.}

The following two algorithms describe the \attacker{} strategy under the restriction of no-splitting. 
In particular, \algref{alg:attacker-init} presents the initial allocation for the \attacker{}, and \algref{alg:attacker-strategy} provides the feedback attacker strategy at time steps $t \geq 0$. 

\begin{algorithm}[ht]
\caption{Attacker Initial Allocation}
\label{alg:attacker-init}
\SetAlgoLined
\SetKwInput{KwInputs}{Inputs}
\KwInputs{Graph $\graph$, \defender{} total resource $X$, \attacker{} total resource $Y$, game horizon $\hmax$\;}
Construct Q-sets using Algorithm~\ref{alg:Q-prop}\;
    \uIf{$\exists k \leq \min\{\hmax, k_\infty\}$ and $i \in \nodeset$ such that $\Delta_X \cap \nsp{k}{i} = \emptyset$}
    {$ k_{\mathrm{min}, -1} \gets \argmin_k \left\{k \leq \min\{\hmax, k_\infty\} \; |\; \Delta_X \cap \nsp{k}{i} = \emptyset \right\}$
    \Comment{{\footnotesize\texttt{find the earliest breach}}}
    \\
    $i_{-1}^* \gets$ any element in  $\{i \, \vert \, \Delta_X \cap \nsp{k_{\mathrm{min}, -1}}{i} = \emptyset$ \}\;}
    \Else{
    $ k_{\mathrm{min}, -1} \gets \infty$ \Comment{{\footnotesize\texttt{no breach found}}} \\
    $i_{-1}^* \gets$ any element in node set $\nodeset$\;
    }
\textbf{Return:} Initial allocation $\yinit = \y^{(i_{-1}^*)}$, guaranteed breach time $k_{\min,-1}$.
\end{algorithm}

\begin{algorithm}[ht]
\caption{Feedback Attacker Strategy}
\label{alg:attacker-strategy}
\SetAlgoLined
\SetKwInput{KwInputs}{Inputs}
\KwInputs{Q-sets, observed \defender{} allocation $\x_{t}$, planning horizon $\hmax$\;}
\uIf{$\exists k \leq \min\{\hmax, k_\infty\}$ and $i \in \neighbor{i_{t-1}}$ such that $\x_t \notin \nsp{k}{i}$}
{$ k_{\mathrm{min}, t} \gets$
\resizebox{.42\textwidth}{!}{$\argmin_k \left\{k \leq \min\{\hmax, k_\infty\} \; |\; \x_t \notin \nsp{k}{i}, \,i \in \neighbor{i_{t-1}} \right\}$}; \Comment{{\footnotesize\texttt{exploit \defender{}'s mistake}}}
\\
$i_{t}^* \gets$ any element in  $\{i \in \neighbor{i_{t-1}} \, \vert \, \x_t \notin \nsp{k_{\mathrm{min}, t}}{i}$ \}\;
}
\Else{
$ k_{\mathrm{min}, t} \gets \infty$\;
$i_{t}^* \gets$ any element in $\neighbor{i_{t-1}}$\;
}
\textbf{Return:} Next allocation $\y_{t} = \y^{(i_t^*)}$, guaranteed breach time $k_{\min,t}$.
\end{algorithm}

%

The \attacker{} can defeat the \defender{} only when the \defender{} allocates resources outside the Q-sets. 
Since we formulated the dDAB game as a game of kind without any performance metric, when the \defender{} allocates resources within $\nsp{k}{i}$, the \defender{} is guaranteed to defend the next $k$ steps, and thus the \attacker{} does not have preference over which node to move to next. 
Therefore, we have arbitrary selections in line~7 of \algref{alg:attacker-init} and line~6 of \algref{alg:attacker-strategy}.
Introducing a cost for the \defender{}'s reallocation is a potential extension of this work. 
Our recent work~\cite{guan2023adversarial} explored this idea and developed a more general framework based on convex body chasing~\cite{friedman1993convex}, where the Q-sets are the convex bodies to be chased.

As we show later in Corollary~\ref{cor:no_splitting}, the \attacker{} has no incentive to split, i.e., if the \attacker{} can win a dDAB game by splitting, it can also win the game without splitting. 
Consequently, the algorithms presented here are sufficient for the \attacker{} to play the dDAB game. 
However, the \defender{} strategies need to be generalized to counter potential splitting \attacker{}, which we will present in the next section.

\subsection{Ring-graph Example}
We apply the proposed algorithms to an example that admits indefinite defense. 
Consider the (directed) ring graph with self-loops shown in~\figref{fig:ring-example}. In this case, the Q-set propagation algorithm converges immediately with $k_\infty \!=\! 0$, implying that $\nsp{\infty}{i} \!=\! \polyreq(\y^{(i)})$ for all nodes $i \!\in\! \nodeset$.
The resulting strategy for the \defender{} is straightforward: upon observing the \attacker{}'s allocation $\y_0$, the \defender{} places one unit of resource at the \attacker{}'s current node and one unit of resource at the node immediately ahead in the ring. This allocation guarantees coverage of both the current and potential next positions of the \attacker{} resource.
The \defender{} can maintain this pattern indefinitely, thereby ensuring indefinite defense.
\begin{figure}[t]
    \centering
    \includegraphics[width=\linewidth]{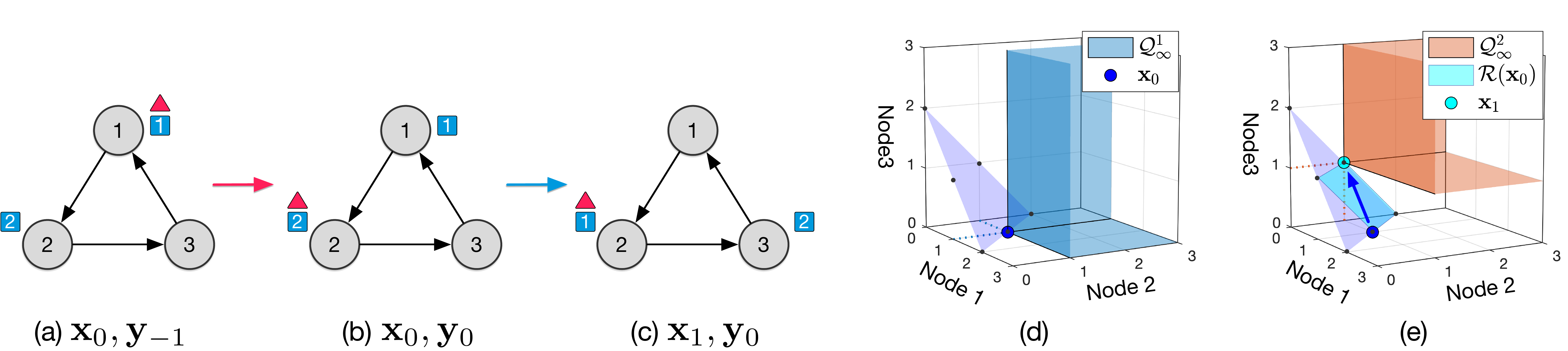}
    \caption{An example of indefinite-defense on a ring graph with self-loops. 
    Subplots (a)-(c) presents an example of an allocation sequence;
    subplots (d) and (e) illustrate the indefinite Q-sets and the construction of the defender strategy.}
    \label{fig:ring-example}
    \vspace{-0.2in}
\end{figure}

Next, we demonstrate how the above \defender{} strategy is generated using \algref{alg:defender-strategy-nsp}. Subplots (d) and (e) in \figref{fig:ring-example} illustrate the computed indefinite Q-sets for nodes 1 and 2 in the ring graph; the Q-set for node 3 is similar, with its vertex located at $[1,0,1]$.

With the \attacker{} robot at node~1 and two \defender{} robots, the initial \defender{} allocation $\x_0$ is obtained by solving the feasibility problem $\Delta_X \cap \nsp{\infty}{1}$, which yields the unique solution $\x_0 = [1,1,0]$, as shown in subplot (d).
After the \attacker{} moves to node 2, the \defender{} computes its next allocation via $\RSet(\x_0) \cap \nsp{\infty}{2}$. This again yields a unique solution, $\x_1 = [0,1,1]$, depicted in subplot (e).
If the \attacker{} remains at node 2, the \defender{}'s allocation remains unchanged. Otherwise, if the \attacker{} moves to node 3 (i.e., $\y_1 = [0,0,1]$), the same process can be repeated, yielding the next \defender{} allocation $\x_2 = [1,0,1]$.
This behavior is consistent with Theorem 3 in the prior work~\cite{shishika2022dynamic},
though here we derive it from an algorithmic framework capable of handling generic graphs.

\begin{remark}
    \label{rmk:discrete-resource}
    Although the Q-sets are constructed based on continuous resources, one can extract policies for indivisible robots by selecting the discrete states lying within the Q-sets (e.g., black dots in \figref{fig:ring-example} (d) and (e)).
\end{remark}

\section{Generalized Defense Strategies}
\label{sec:split}
This section generalizes the \defender{} strategy in the previous section to scenarios where the \attacker{} can split its resources to multiple nodes. 
In particular, we show that if the \defender{} has sufficient amount of resources to defend against any no-splitting \attacker{} strategy, then it can defend against any \attacker{} strategy, including the splitting ones. 
This result implies that the \attacker{} can win the game if and only if it can win using a no-splitting strategy, and consequently the \attacker{} does not have any incentive to split its resources to win the game.
Finally, we obtain the critical resource ratio~(CRR), which describes the necessary and sufficient amount of the \defender{} resource required to guarantee defense against any \attacker{} strategy.   

\subsection{Attacker and Defender Subteams}
To extend the analysis from no-splitting strategies to more general strategies, we introduce the notion of subteams.
\begin{definition}[Attacker Subteam]
We refer to the \attacker{} resource allocated to each node as an  \emph{\attacker{} subteam}.
The size of the $i$-th \attacker{} subteam (on node $i$) at time $t$ is $[\y_t]_i$.
\end{definition}

In general, any \attacker{} action can be viewed as a superposition of the subteam actions, which results in the splitting and merging of subteams into a new set of subteams.
\figref{fig:subteam_attacker_splitting} illustrates an example where two \attacker{} subteams split and merge into a new set of three subteams.
Note that the \attacker{}'s action to achieve the allocation in \figref{fig:subteam_attacker_splitting}(d) from \figref{fig:subteam_attacker_splitting}(a) is non-unique.
\begin{figure}[ht]
    \centering
    \includegraphics[width = 0.8\textwidth]{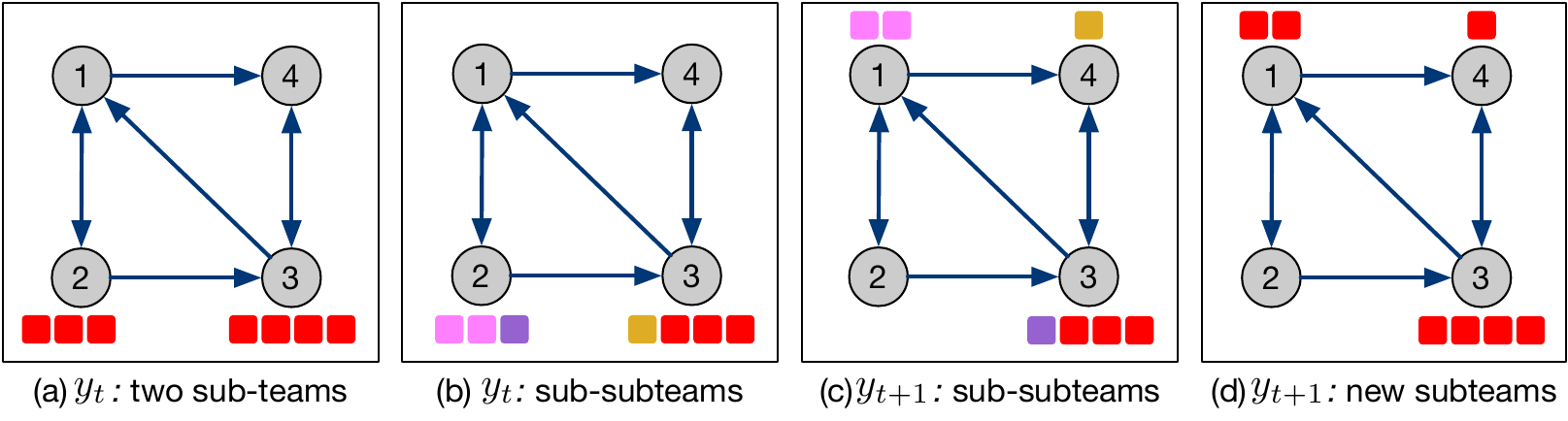}
    \vspace{-0.1in}
    \caption{Splitting and merging of an \attacker{} subteam. 
    (a) Two \attacker{} subteams at $t$: a 3-unit subteam on node 2 and a 4-unit subteam on node 3. 
    (b) Each subteam is about to split into two: magenta-purple and yellow-red (the game is still at time $t$.)
    (c) The re-allocation of each subteam after $t$. 
    (d) Three new subteams on nodes 1,3 and 4 at $t+1$.}
    \label{fig:subteam_attacker_splitting}
    \vspace{-0.15in}
\end{figure}

Based on the necessity and sufficiency of Q-sets, we define the $i$-th \defender{} subteam as the subset of the \defender{} resource that can defend against the $i$-th \attacker{} subteam, assuming that the \attacker{} subteam does \emph{not} further split in the future.
\begin{definition}[Defender Subteam]\label{def:defender_subteam}
The $i$-th (scaled) \defender{} subteam is defined as
\begin{equation}\label{eq:defender_subteam}
    \x_{t}^{(i)} \triangleq \frac{[\y_{t-1}]_i}{Y} \xhat_t^{(i)}, 
    \text{\;where\;} \xhat_t^{(i)}\in\nsp{k}{i}.
\end{equation}
We refer to $\xhat_t^{(i)}$ as the unscaled \defender{} subteam. 
\end{definition}

Note that the Q-sets in~\eqref{eqn:Q-set-def} are defined based on the total \attacker{} resource $Y$. 
Consequently, each \defender{} subteam is scaled according to the size of its corresponding \attacker{} subteam in \eqref{eq:defender_subteam}. 
The defense condition for the \defender{} subteams is thus linked to the Q-sets through the unscaled \defender{} subteam $\xhat$.

As illustrated in~\figref{fig:splitting-attacker}(a), the \attacker{} subteams at nodes~1 and 2 have sizes of 1 (red robot) and 2 (dark red and pink robots), respectively. 
The initial \defender{} subteams are $\x_{0}^{(1)} = [1,1,0]$ and $\x_{0}^{(2)} = [0,2,2]$.
Specifically, the 1st \defender{} subteam consists of the two blue robots, while the 2nd \defender{} subteam comprises the dark blue and cyan robots.
The dark blue \defender{} robots are assigned to defend against the dark red \attacker{} robot, and the cyan robots are designated to the pink robot.

Based on the results from the no-splitting \attacker{} scenario, if the \attacker{} subteams do not split further, each \defender{} subteam can successfully defend against its respective \attacker{} subteam for the next $k$ steps, ensured by condition~\eqref{eq:defender_subteam}.

\begin{figure}[b]
    \centering
    \includegraphics[width = 0.65\textwidth]{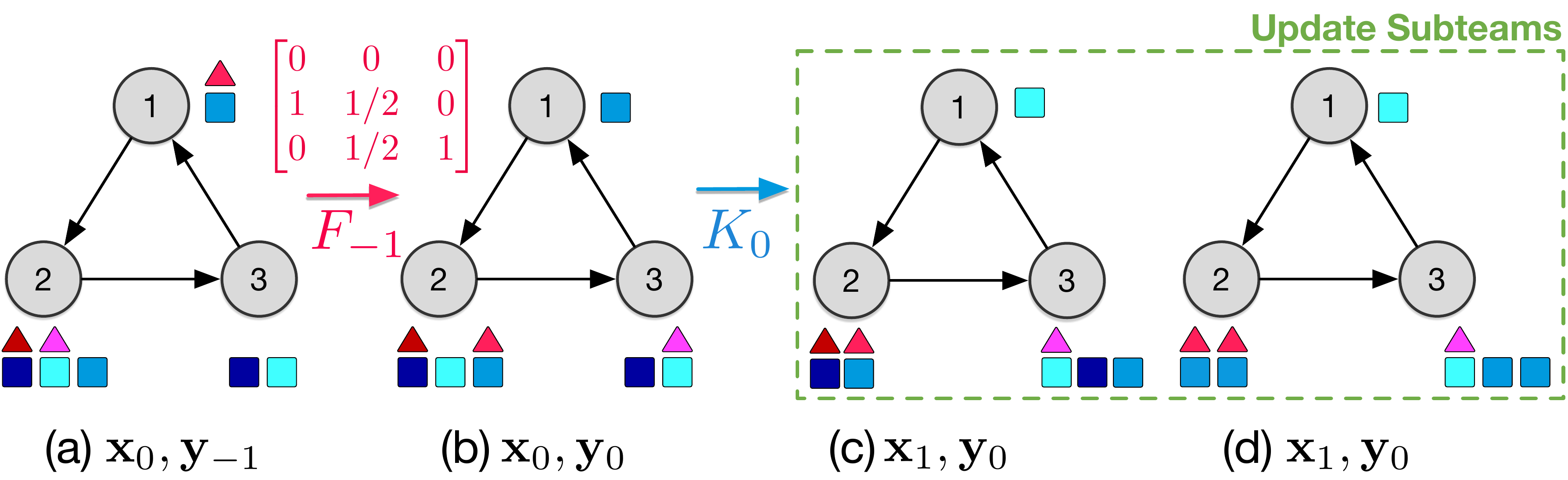}
    \caption{An illustrative ring-graph example with splitting \attacker{} resources and the response of corresponding \defender{} subteam.}
    \label{fig:splitting-attacker}
\end{figure}

\subsection{Generalization to Splitting Attacker}
We now extend the above construction to scenarios where the \attacker{} subteams further split, which is summarized in the following theorem.
\begin{restatable}{theorem}{mt}
    \label{thm:kstep_defense_general_new}
    Given the \attacker{}'s initial state $\yinit$, the \defender{} can defend against any \attacker{} strategy until time~$T$ if the \defender{}'s initial state $\x_0$ can be expressed as 
    \begin{equation}
        \label{eqn:splitting-defender-allocation}
        \x_0 = \sum_{i\in \nodeset} \frac{[\yinit]_i}{Y} \hat{\x}_0^{(i)}, ~\text{for some } \hat{\x}_0^{(i)} \in \nsp{T}{i}.
    \end{equation}
\end{restatable}

\begin{proof}
A formal proof is presented in Appendix~\ref{appdx-sec:general-defender-thm}. 
\end{proof}

For the rest of this subsection, we present the basic intuition behind the proof construction. 

Consider a generic \attacker{} action $F_{t-1}$ that transitions $\y_{t-1}$ to $\y_t$. 
Let $\bs{f}^{(i)}$ represent the $i$-th column of $F_{t-1}$, i.e., $F_{t-1} = [\bs{f}^{(1)}, \bs{f}^{(2)}, ..., \bs{f}^{(N)}]$.
The vector $\bs{f}^{(i)}$ encodes the splitting action of the \attacker{} subteam on node $i$,
where the fraction of this subteam relocating to node $j$ is given by $[\bs{f}^{(i)}]_j$.
For instance, in \figref{fig:splitting-attacker}, the red action $F_{-1}$ yields $[\bs{f}^{(2)}]=[0,1/2, 1/2]^\top$, meaning that one of the two red robots on node 2 remains, while the other moves to node 3.

The $i$-th \defender{} subteam should react to the splitting of the $i$-th \attacker{} subteam in the following manner:
\begin{enumerate}
    \item The $i$-th \defender{} subteam is divided into ``sub-subteams", according to the splitting action $\bs{f}^{(i)}$ of the $i$-th \attacker{} subteam.
    \item Each $j$-th \defender{} sub-subteam of the $i$-th subteam counteracts the $j$-th \attacker{} sub-subteam that moves from node $i$ to node $j$.
    Specifically, the sub-subteam is given by 
    \[
        \x_{t}^{(i \to j)} = [\bs{f}^{(i)}]_j  \x^{(i)}_t = \frac{[\bs{f}^{(i)}]_j[\y_{t-1}]_i}{Y}  \xhat_{t}^{(i)}.
    \]
    \item This counteraction is executed by applying the action $K^{(i \to j)}$, resulting in a new unscaled defender subteam:
    \[\xhat_{t+1}^{(i\to j)} = K^{(i\to j)}\xhat_{t}^{(i)} \in \nsp{k-1}{j}.\]
    \item The new state of the \defender{} sub-subteam at the next time step is then
    \begin{equation}
        \label{eqn:sub-subteam-requirement}
        \x_{t+1}^{(i\to j)} = K^{(i\to j)} \x_{t}^{(i \to j)} = \frac{[\bs{f}^{(i)}]_j[\y_{t-1}]_i}{Y}\xhat_{t+1}^{(i\to j)}.
    \end{equation}    
\end{enumerate}

In the example shown in \figref{fig:splitting-attacker}, the new states of the sub-subteams are given by $\x_{1}^{(1\to 2)} = [0,1,1]$ (blue), $\x_{1}^{(2\to 2)} = [0,1,1]$ (dark blue), and $\x_{1}^{(2\to 3)} = [1,0,1]$ (cyan).
Notably, $\x_{t+1}^{(i\to j)}$ contributes only a portion of the new $j$-th \defender{} subteam, originating from the previous $i$-th subteam.

To compute the new $j$-th \defender{} subteam, we group the \defender{} resources from different subteams that responded to the \attacker{} resources and ended up at node $j$. This grouping is expressed as:
\begin{equation}
    \label{eqn:new-def-subteam-j}
    \x_{t+1}^{(j)} = \sum_{i\in \nodeset} \x_{t+1}^{(i\to j)} = \sum_{i\in \nodeset} \frac{[\bs{f}^{(i)}]_j[\y_{t-1}]_i}{Y} \xhat_{t+1}^{(i\to j)}.
\end{equation}
This step is illustrated in \figref{fig:splitting-attacker}(d), where the new \defender{} subteams are $\x_{1}^{(2)} = [0,2,2]$ (blue), and $\x_{1}^{(3)} = [1,0,1]$ (cyan). Meanwhile, two new \attacker{} subteams form at nodes 2 and 3, with sizes of 2 and 1, respectively.

The critical question is whether this $j$-th new \defender{} subteam can effectively defend against the new \attacker{} subteam at node $j$. 
Unscaling the new $j$-th \defender{} subteam in~\eqref{eqn:new-def-subteam-j} gives 
\begin{equation}
    \xhat_{t+1}^{(j)} = \frac{Y}{[\y_t]_j} \x_{t+1}^{(j)} = \sum_i \frac{[\bs{f}^{(i)}]_j[\y_{t-1}]_i}{[\y_t]_j} \xhat_{t+1}^{(i\to j)}.%
\end{equation}
Note that the size of the new \attacker{} subteam on node $j$ is $[\y_{t}]_j = \sum_i [\bs{f}^{(i)}]_j [\y_{t-1}]_i$.
Thus, $\xhat_{t+1}^{(j)}$ is a convex combination of the states ${\xhat_{t+1}^{(i\to j)}}$, all of which are in the Q-set $\nsp{k-1}{j}$ as constructed in~\eqref{eqn:sub-subteam-requirement}. Given that the Q-sets are polytopes (Theorem~\ref{thm:Q-polytope}), it follows that $\xhat_{t+1}^{(j)} \in \nsp{k-1}{j}$.
Therefore, we can prove by induction that the \defender{} subteams at the next time step can maintain their defense for an additional $k\!-\!1$ steps.

\subsection{Generalized Defender Strategy}
As a direct consequence of Theorem~\ref{thm:kstep_defense_general_new}, the \defender{} strategy outlined in~\secref{sec:nsp} can be extended to scenarios in which the \attacker{} distributes its resources across multiple nodes.
The generalized \defender{} strategy is summarized in Algorithms~\ref{alg:initial-defender-allocation} and~\ref{alg:defender-strategy}.
Algorithm~\ref{alg:initial-defender-allocation} determines the initial \defender{} allocation, while Algorithm~\ref{alg:defender-strategy} leverages the subteam construct introduced above to respond to a splitting \attacker{}.

In particular, Lines 3–12 of Algorithm~\ref{alg:initial-defender-allocation} aim to improve \defender{} performance when the \attacker{} does not concentrate all resources at the most critical node leading to the earliest breach.

In Algorithm~\ref{alg:defender-strategy}, Line 2 attempts to exploit potential \attacker{} mistakes.
Lines 3–9 respond to the splitting of \attacker{} resources, and Lines 10–13 construct the new subteam and output the next desired \defender{} allocation.

\begin{algorithm}[h]
\caption{Initial Defender Allocation}
\label{alg:initial-defender-allocation}
\SetAlgoLined
\SetKwInput{KwInputs}{Inputs}
\KwInputs{Graph $\graph$, total resources $X$ and $Y$, \attacker{} initial allocation $\y_{-1}$, game horizon $\hmax$\;}
Construct Q-sets via~\algref{alg:Q-prop}\;
$\mathcal{I}_{-1} \gets \{i \, | \, [\y_{-1}]_i > 0\}$\;
\For {$k = 1, \ldots, T$}{
\For{$i \in \mathcal{I}_{-1}$}{
    $\xhat^{(i)}_{0,k}\gets \argmin_{\x \in \Delta_X \cap \nsp{i}{k}} \mathbf{1}^\top \x$ \;
}
$X_{\mathrm{tot},k} = \sum_{i \in \mathcal{I}_{-1}} \frac{[\y_{-1}]_i}{Y} \mathbf{1}^\top \xhat^{(i)}_0$ \;
\If{$X_{\mathrm{tot},k} > X$}{
$k_{\mathrm{max},0} = k-1$ \;
\textbf{Break}\;
}

}
 $X_{\mathrm{tot}} \gets X_{\mathrm{tot},k_{\mathrm{max},0}}$
 and 
 $\xhat^{(i)}_{0} \gets \frac{X}{X_{\mathrm{tot}}} \xhat^{(i)}_{k_{\mathrm{max},0}}$

$\x_0\gets \sum \frac{[\y_{-1}]_i }{Y}\xhat^{(i)}_{0}$;

\textbf{Return:} Q-sets $\{\nsp{k}{i}\}_{k, i}$, initial allocation and subteams $(\x_{0}, \{\xhat_{0}^{(i)}\}_i)$, longest guaranteed defense time $k_{\max,0}$.
\end{algorithm}

\begin{algorithm}[h]
\caption{Feedback Defender Strategy}
\label{alg:defender-strategy}
\SetAlgoLined
\SetKwInput{KwInputs}{Inputs}
\KwInputs{Q-sets, previous scaled subteams $\{\xhat_{t-1}^{(i)}\}_i$, observed \attacker{} allocations $\y_{t-1}$ and  $\y_{t-2}$\;}
$\mathcal{I}_{t-1} \gets \{i \, | \, [\y_{t-1}]_i > 0\}$\;
$ k_{\mathrm{max}, t} \gets \argmax_k \left\{k \; |\; 
\RSet(\hat{\x}_{t-1}^{(i)})\cap \nsp{k}{i} \ne \emptyset \; \forall \; i \in \mathcal{I}_{t-1} \right\}$;
 ~~~\Comment{find the longest survival time}
\\
    {\footnotesize \tcp{Generate \defender{} subteams' actions}}
\For{$i \in \nodeset$ \st $\xhat_{t-1}^{(i)}\neq \bs{0}$}{

    $\bs{f}_i \gets $ $i$-th column of $F_{t-2}$;
    
    {\footnotesize \tcp{Reaction to \attacker{} subteam i splitting}}
    \For{$j \in \nodeset$ \st $[\bs{f}_i]_j\neq 0$}{
    $\xhat^{(i \to j)}_{t}\gets$ element in $\RSet{(\x_{t-1}^{(i)})} \cap \nsp{k_{\max,t}}{j}$;
    
    $K^{(i\to j)} \gets $ action s.t. $\xhat^{(i \to j)}_{t} = K^{(i\to j)} \xhat_{t-1}^{(i)}$;
    }
    
    
}

{\footnotesize \tcp{New subteams}}
\For{$i \in \nodeset$}{
$\x_{t}^{(i)} \gets \sum_j \frac{[\y_{t-2}]_j}{Y}[\bs{f}_j]_i\xhat^{(j \to i)}_{t}$\footnotemark;
}
$\x_{t} = \sum_i \x_{t}^{(i)}$;

$K_t \gets $ overall action from \eqref{eqn:overal-action-formula}\;

\textbf{Return:} Action $K_t$, resultant subteams $\{\xhat_{t+1}^{(i)}\}_i$,  guaranteed defense time $k_{\max,t}$.
\end{algorithm}
\vspace{-0.1in}

\newpage
\subsection{The Critical Resource Ratio}
\label{sec:CRR}
\footnotetext{
Note that $\x_{t}^{(i)}= \frac{[\y_{t-1}]_i}{Y} \xhat_{t}^{(i)} = \frac{[\y_{t-1}]_i}{Y} \sum_j \frac{[\y_{t-2}]_j}{[\y_{t-1}]_i}[\bs{f}_j]_i\xhat^{(j \to i)}_{t} = \sum_j \frac{[\y_{t-2}]_j}{Y}[\bs{f}_j]_i\xhat^{(j \to i)}_{t}$, for $i$ such that $[\y_{t-1}]_i > 0$. If $[\y_{t-1}]_i = 0$, then there is no need to create a subteam designated for node $i$.
}

We leverage the results in the previous sections to identify the critical resource ratio (\multfactor{}, see Definition~\ref{def:crr}).
In Section~\ref{sec:nsp}, we have shown that being in $\nsp{k}{i}$ is the necessary and sufficient for the \defender{} to defend against any no-splitting \attacker{} strategy that starts from $\yinit = \y^{(i)}$ for $k$ steps.
This leads to an intermediate version of the \multfactor{} defined for the case of no-splitting \attacker{} starting on node $i$:
\begin{equation}
    \label{eqn:crr-def-nsp}
    \nspfactornode{k}{i} \triangleq \min_{\x \in \nsp{k}{i}} \ones^\top \x.
\end{equation}
Given that the \attacker{} can freely select its initial state $\yinit$, we define the k-step CRR given a no-splitting \attacker{} as
\begin{equation}
    \nspfactor{k} \triangleq \max_{i \in \nodeset} \nspfactornode{k}{i}.
\end{equation}

The following result shows that the CRR against general splitting \attacker{} is identical to the one defined under no-splitting restriction.

\begin{theorem}[Critical Resource Ratio]
\label{thm:sufficient_resource}
The necessary and sufficient resource ratio for the \defender{} to  achieve $k$-step defense against any  \attacker{} strategy  is given by
\begin{equation}
    \factor{k}{} = \nspfactor{k} {= \frac{1}{Y}\max_{i \in \nodeset} \min_{\x \in \nsp{k}{i}} \ones^\top \x.}
\end{equation}
\end{theorem}
\begin{proof}
It is obvious that $\factor{k}{}\geq\nspfactor{k}$, since it is necessary to guard against no-splitting  strategies. 
Consequently, it suffices to show that $X=\nspfactor{k}{}Y$ is sufficient to guard against any admissible \attacker{} strategy, including the ones with splitting.

Using the result of Theorem~\ref{thm:kstep_defense_general} with $t=0$ and $T=k$, consider the following initial \defender{} state that is sufficient to guard against any given $\y_{-1}$ over the next $k$ time steps:
\begin{equation}
    \x_0 = \frac{1}{Y}\sum_i [\y_{-1}]_i\, \xhat_{0}^{(i)}, \text{\;\;where\;} \xhat_{0}^{(i)}\in\nsp{k}{i}.
\end{equation}

The minimum amount of resource required to achieve the above allocation is given by
\begin{align*}
    &\min_{(\xhat_{0}^{(1)}, \ldots, \xhat_{0}^{(|\nodeset|)}) \in \nsp{k}{1} \times \cdots \times \nsp{k}{|\nodeset|}} \bs{1}^\top \Big(\frac{1}{Y}\sum_i [\y_{-1}]_i\, \xhat_{0}^{(i)}\Big) =\\
    &=    \frac{1}{Y}\sum_i [\y_{-1}]_i\, \Big(\min_{\xhat_{0}^{(i)}\in\nsp{k}{i}} \bs{1}^\top\xhat_{0}^{(i)}\Big)
    =\frac{1}{Y}\sum_i [\y_{-1}]_i \, \nspfactornode{k}{i} X \\
    &\leq \frac{1}{Y}\sum_i [\y_{-1}]_i \nspfactor{k} X = \nspfactor{k} X,
\end{align*}
where the equality is given when the \attacker{} initial state $\yinit$ places all its resources on the node $i^* \in \argmax_{i \in \nodeset} \nspfactornode{k}{i}$.
Hence, we have $\factor{k}{} = \nspfactor{k}$. 
\end{proof}

\begin{corollary}[No Incentive to Split]
\label{cor:no_splitting}
For a given graph $\graph$ and resources $X$ and $Y$, the \attacker{} has a strategy to win the dDAB game if and only if it has a no-splitting winning strategy.
\end{corollary}
\begin{proof}
The sufficiency is given trivially. 
The necessity comes as a direct consequence of Theorem~\ref{thm:sufficient_resource}.
If the \attacker{} can win a dDAB game with some strategy, we have that $Y \geq \factor{k}{} X$. 
From Theorem~\ref{thm:sufficient_resource}, we obtain $Y \geq \factor{k}{} X = \nspfactor{k} X$, which implies that the \attacker{} can also win with a no-splitting strategy.

\end{proof}

The following corollary regarding the indefinite defense is a direct consequence of Corollary~\ref{cor:no_splitting}.
\begin{corollary}
    If an indefinite defense is feasible against all no-splitting \attacker{} strategies, the \defender{} can also indefinitely defend against any \attacker{} strategies.
\end{corollary}

\begin{proof}
    Note that if an \attacker{} can win the game with a splitting strategy, it must breach a node at some finite time step. 
    From Corollary~\ref{cor:no_splitting}, it can also win without splitting, which contradicts the assumption.
    
\end{proof}

\subsection{Additional Properties of CRR}
We present some additional properties of \multfactor{} that are straightforward to obtain.
\begin{proposition}
    The sequence $(\factor{T}{})_{T=1}^\infty$ is monotonically nondecreasing with respect to horizon $T$.
\end{proposition}
This property is obvious from the fact that the ability to defend over $T$ time steps immediately implies the ability to defend any duration less than $T$.
\footnotetext{
Note that $\x_{t}^{(i)}= \frac{[\y_{t-1}]_i}{Y} \xhat_{t}^{(i)} = \frac{[\y_{t-1}]_i}{Y} \sum_j \frac{[\y_{t-2}]_j}{[\y_{t-1}]_i}[\bs{f}_j]_i\xhat^{(j \to i)}_{t} = \sum_j \frac{[\y_{t-2}]_j}{Y}[\bs{f}_j]_i\xhat^{(j \to i)}_{t}$, for $i$ such that $[\y_{t-1}]_i > 0$. If $[\y_{t-1}]_i = 0$, then there is no need to create a subteam designated for node $i$.
}

\vspace{+0.1in}

\begin{proposition}
    If $\nsp{\infty}{i}\ne \emptyset$ for some $i \in \nodeset$, then $\nspfactornode{\infty}{i} = \nspfactornode{\infty}{j} < \infty$ for all $i,j \in \nodeset$.
\end{proposition}

\vspace{-0.2in}
\begin{proof}
    Note that for any finite $k$, $\nspfactornode{k+1}{i} \geq \nspfactornode{k}{j}$ for all $j \in \neighbor{i}$, since for a \defender{} to defend an \attacker{} starting from node $i$ for $k+1$ steps, it has to be able to defend $k$ steps after the \attacker{} moves to node $j$.
    Consequently, we have $\nspfactornode{\infty}{i} \geq \nspfactornode{\infty}{j}$ for all $j \in \neighbor{i}$. 
    Since the graph is strongly connected, there exists a directed path from $j$ to $i$. 
    One can then cascade the inequality along the path from $j$ to $i$, and it follows that $\nspfactornode{\infty}{j} \geq \nspfactornode{\infty}{i}$. 
\end{proof}

\begin{proposition}[Lowerbound of $\factor{T}{}$~\cite{shishika2022dynamic}]
For a general graph and an arbitrary $T$,  $\factor{T}{}$ is bounded from below by $\underline{\alpha} = d^+_{\mathrm{max}}$, where $d^+_{\mathrm{max}}=\max_{j\in\nodeset}d^+_j$ is the maximum out-degree of the graph. 
\end{proposition}

This property can be proved by considering the case where the \attacker{} initially concentrates all its resources at the node with the maximum out degree. Unless the \defender{} allocates an equal amount or more to every one of the neighboring nodes, the \attacker{} will have an action to win the game, i.e., move all the \attacker{} resources to the neighboring node where the \defender{} allocates less than $Y$ unit of resources.
\vspace{+0.1in}

\begin{proposition}[Upperbound of $\factor{T}{}$]
For a strongly-connected graph and an arbitrary $T\in[0, \infty]$, $\factor{T}{}$ is bounded from above by $\bar{\alpha} = \sum_{i\in \nodeset} L_i$, where $L_i$ is the length of the shortest loop that passes through node~$i$.
\end{proposition}
\begin{proof}
    Since the graph is strongly connected, for every node~$i$, there is a loop that passes through $i$. 
    Then, for every node~$i$, the \defender{} can have $L_i Y$ unit of resource patrolling the shortest loop that passes through $i$, resulting every node on the loop (in particular, node $i$) having $Y$ unit of \defender{} resource at all time.
\end{proof}

\section{Algorithmic Solution}
\label{sec:algorithm}
In this section, we first develop an algorithm to numerically construct the Q-sets. 
The proposed algorithm also helps us prove Theorem~\ref{thm:Q-polytope}, which states that all Q-sets are polytopes.
Recall that we treated the next state $\x_{t+1}$ in the reachable set as the action for the \defender{} to take at time step $t$. 
In reality, however, the \defender{} needs to find a feasible action $K_t \in \K$ to reach $\x_{t+1}$.
In the second subsection, we formulate this action extraction problem as a linear program, which can be solved efficiently. 

\subsection{Q-set Construction}
Recall the recursive definition of the Q-sets in~\eqref{eqn:Q-set-def-k}:
\begin{equation*}
    \nsp{k}{i} = \Big\{\x \; \big\vert \; \x \in \polyreq(\y^{(i)}) \and \RSet(\x) \cap \nsp{k-1}{\ip} \ne \emptyset ~\forall \ip \in \neighbor{i} \Big\}.
\end{equation*}
To numerically construct the Q-sets, we first examine the properties of the set $\{\x | \RSet(\x) \cap \nsp{k}{\ip} \ne \emptyset\}$.
This set consists of states from which the \defender{} can reach $\nsp{k}{\ip}$ at the next time step.
It is unclear yet, whether these ``inverse reachable sets" induce nice properties for the Q-sets. 

We formally define the inverse reachable set of some set $P$ as follow:
\begin{definition}[Inverse Reachable Set]
Given a set $P \subseteq \mathbb{R}^N_{\geq 0}$, we define the inverse reachable set of $P$ as 
\begin{equation}
    \RsetInv(P) = \left\{\x \;\big\vert\; \RSet(\x) \cap P \ne \emptyset \right\}.
\end{equation}
\end{definition}

With the notion of the inverse reachable set, we can simplify the recursive construction of Q-sets in~\eqref{eqn:Q-set-def} as 
\begin{subequations}
\begin{alignat}{2}
    \nsp{0}{i} &= \polyreq(\y^{(i)}), \\
    \nsp{k}{i} &= \Big(\bigcap_{\ip\in \neighbor{i}} \RsetInv(\nsp{k-1}{\ip})\Big) \; \cap \; \polyreq(\y^{(i)}) \quad & \forall ~ k \geq 1.
\end{alignat}
\end{subequations}

We now discuss the computation of the inverse reachable sets. 
Note that any admissible action $K\in \K$ can be reversed.
That is, if one can use an action to reach $\x_{t+1}$ from some $\x_t$, then one can also find a reverse action that brings the \defender{}'s allocation from $\x_{t+1}$ to $\x_t$. 
Based on this intuition, we introduce the notion of a reversed graph, which has the same node set as the original graph but with all the directed edges reversed.

\begin{definition}
For a graph $\graph=(\nodeset, \edgeset)$ with connectivity matrix $A$,
its reversed graph $\reversedgraph = (\nodeset, \revcerseedgeset)$ is defined based on the connectivity matrix $\widetilde{A} = A ^\top$.
\end{definition}

We denote $\widetilde{\mathcal{K}}$ as the admissible action set of $\widetilde{\mathcal{G}}$.
The reachable set for the reversed graph is then defined as 
\begin{equation*}
    \widetilde{\RSet}(\x) = \left\{\x'\;|\; \exists \widetilde{K} \in \widetilde{\K} \st \x'=\widetilde{K}\x\right\}.
\end{equation*}

The following lemma relates the (forward) actions and the reversed actions.

\begin{restatable}{lemma}{ia}
\label{lmm:inverse-action}
Given an arbitrary admissible action under the original graph $K \in \K$ and an arbitrary starting state $\x$, 
suppose the resultant state is $\x' = K \x$.
We can reverse the action using an admissible action under the reversed graph $\widetilde{K} \in \widetilde{\K}$ to achieve $\x = \widetilde{K} \x'$.
The reverse action $\widetilde{K}$ can be constructed as 
\begin{equation}
    \label{eqn:reverse-action}
    \left[\widetilde{K}\right]_{ij} = 
    \begin{cases}
        \frac{[K]_{ji} [\x]_i}{[\x']_j} &\text{if } [\x']_j > 0, \\
        \frac{1}{\sum_{i} \left[\widetilde{A}\right]_{ij}} &\text{if } [\x']_j = 0 \text{ and } [\widetilde{A}]_{ij} = 1,
        \\
       0 &\text{if } [\x']_j = 0 \text{ and } [\widetilde{A}]_{ij} = 0.
    \end{cases}
\end{equation}
\end{restatable}
\begin{proof}
    See Appendix~\ref{appdx-sec:reverse-graph}.
\end{proof}

Based on the above result, we have the equivalence between the inverse reachable set of the graph $\graph$ and the reachable set of the reversed graph $\reversedgraph$.

\begin{restatable}{lemma}{irs}
\label{lmm:inverse-rset}
For any graph $\graph$, we have
\begin{equation*}
    \RsetInv(P) = \widetilde{\RSet}(P) \qquad \forall \; P \subseteq \R_{\geq0}^N.
\end{equation*}
\end{restatable}
\begin{proof}
    See Appendix~\ref{appdx-sec:reverse-graph}.
\end{proof}

Lemma~\ref{lmm:inverse-rset} leads directly to the following computationally-friendly definition of the Q-sets:
\begin{subequations}
    \label{eqn:q-set-comp}
    \begin{alignat}{2}
        \nsp{0}{i} &= \polyreq(\y^{(i)}), \\
        \nsp{k}{i} &= \Big(\bigcap_{\ip\in \neighbor{i}} \widetilde{\RSet}(\nsp{k-1}{\ip})\Big) \; \cap \; \polyreq(\y^{(i)}) \quad & \forall ~ k \geq 1.
    \end{alignat}
\end{subequations}

Based on the above result, we can easily prove that the Q-sets are polytopes.

\py*
\begin{proof}
    Recall from Lemma~\ref{lmm:RSet-poly} that the reachable set of a polytope remains a polytope. Given the recursive construction of the Q-sets in~\eqref{eqn:q-set-comp}, and noting that all operations (reachable set computations and intersections) preserve polyhedrality, the result follows by induction on $k$. 
\end{proof}

As a direct consequence of Q-sets being polytopes, we have the following corollary.
\begin{corollary}
    The k-step CRR, $\factor{k}{}$, is attained at one of the vertices of the Q-set.
\end{corollary}
\begin{proof}
    Since the Q-sets are polytopes, the optimization for $\nspfactornode{k}{i}$ is a linear program as in~\eqref{eqn:crr-def-nsp}.
    Furthermore, since the CRR is bounded from below by zero, an optimal solution is attainable and can be attained on one of the vertices. 
\end{proof}

\subsection{Action Extraction From Q-Sets}
Recall that we have treated the next state $\x_{t+1}$ in the reachable set as the action chosen by the \defender{} at step~$t$. 
Executing such a strategy in practice requires identifying a feasible action matrix $K_t \in \K$ that satisfies $\x_{t+1} = K_t \x_t$.

By construction of the reachable set $\RSet(\x_t)$, the existence of such an action is guaranteed. 
We propose to find a feasible action $K_t$ that transitions the system to $\x_{t+1}$ by solving a simple matrix equation.
Recall the characterization of $\RSet(\x_t)$ in Lemma~\ref{lmm:RSet-Polytope}.
Specifically, if $\x_{t+1} \in \RSet(\x_t)$, it satisfies
\begin{align}\label{eqn:back-prop-coeff}
    \x_{t+1} = \sum_{\ell} \lambda^{(\ell)} \left(\hatK^{(\ell)}\x_t \right).
\end{align}

Any $\x_{t+1}$ in the intersection would have the safety guarantee, and consequently the selection can be arbitrary. 
Since the intersection set is a bounded polytope,
one may simply select the centroid or a vertex of the intersection  as $\x_{t+1}$.
Since $\x_t$ and $\{\hatK^{(\ell)}\}_\ell$ are all known variables at time step $t$, the vector-form coefficients $\bm{\lambda}$ can be found by solving the following problem:
\begin{equation}
    \label{eqn:action-extraction}
     \Phi \bm{\lambda} = \x_{t+1}, \st \bm{\lambda} \geq 0 \text{ and } \sum_{\ell} \lambda^{(\ell)} = 1,
\end{equation}
where the matrix $\Phi \in \mathbb{R}^{|\nodeset| \times |\hatcalK|}$ has $\left(\hatK^{(\ell)}\x_t\right)_{\ell = 1}^{|\hatcalK|}$ as its columns.
Again, the feasibility of~\eqref{eqn:action-extraction} is guaranteed, due to the construction of $\RSet(\x_t)$.
With the solved $\bm{\lambda}$, the feasible action that brings the \defender{} from $\x_t$ to $\x_{t+1}$ is given by 
\begin{equation*}
    K_{t} = \sum_{\ell} \lambda^{(\ell)} \hatK^{(\ell)}.
\end{equation*}

\section{Numerical Illustrations}

This section provides numerical examples that illustrate the results developed in the previous sections.

\subsection{Q-set Propagation}
\noindent
\figref{fig:qset_propagation} illustrates how the Q-sets, $\nsp{k}{i}$, change with the horizon $k$.
For the three-node graph selected for this example, the propagation in~\algref{alg:Q-prop} converges after four iterations, at which point the algorithm finds that $k_\infty= 4$. The CRR for this graph is $\factor{\infty}{}=3$.
\begin{figure*}[b]
    \centering    \includegraphics[width=\textwidth]{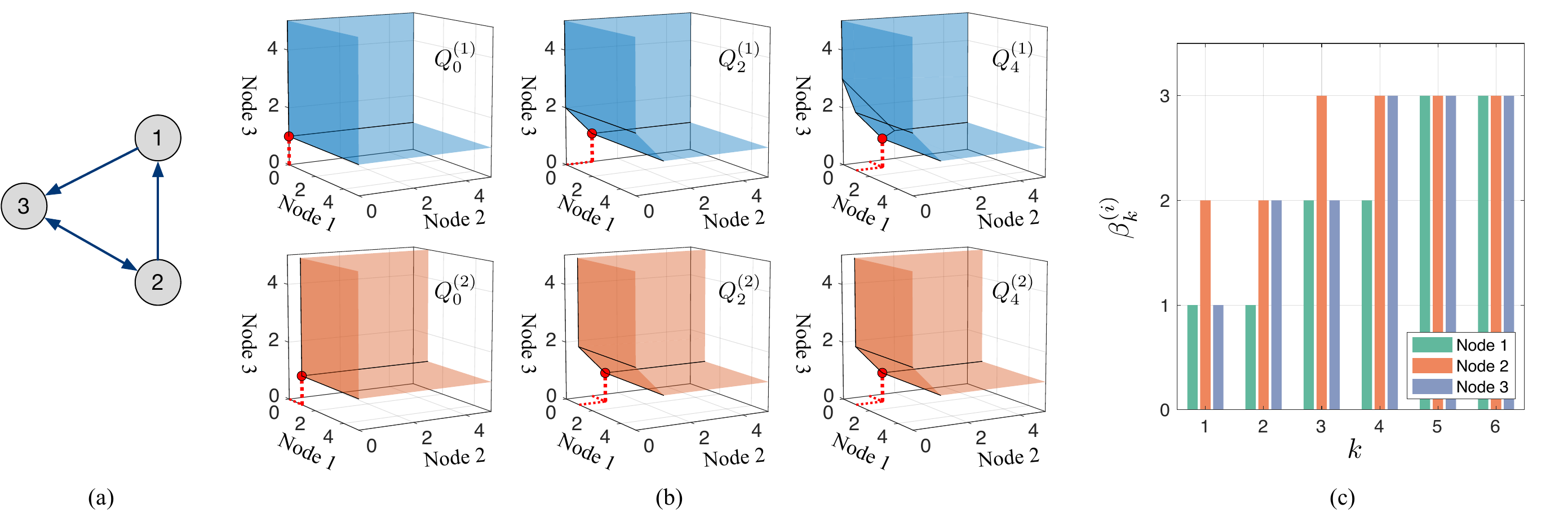}
    \caption{Illustration of Q-sets evolution with different horizons $k$. (a) The three node graph used for this example. (b) We consider a single unit of \attacker{} resource, and $\nsp{k}{i}$ for nodes $i=\{1,2\}$ and horizon $k=\{0,2,4\}$ are shown here for brevity. 
    The red dot in each figure indicates the element in the Q-set that achieves the smallest amount of resource for $k$-step defense, i.e., $\nspfactornode{k}{i}$ in~\eqref{eqn:crr-def-nsp}. (c) The evolution of $\nspfactornode{k}{i}$ on each node, until they converge at $k_\infty = 6$. 
    }
    \label{fig:qset_propagation}
\end{figure*}
It is worth noting that the Q-set for a given node may not change at every time step: e.g., $\nsp{k}{1}$ changes only twice between $k=1$ to $2$ and between $k=3$ to $4$.

We can verify the monotonicity of the Q-sets described in Remark~\ref{rem:monotonicity} by observing how the Q-sets get ``carved off'' and become smaller as $k$ increases.
Specifically, some regions of the state space with small amount of resources get excluded when $k$ changes from $0$ to $2$ and similarly from $2$ to $4$.
As an example, the state $\x = [0,0,1] \in \nsp{0}{1}$ can guard against any immediate next action made by a unit \attacker{} at node 1 (i.e., $\y^{(1)}$). This is shown by the red dot in $\nsp{0}{1}$ (top left subfigure in \figref{fig:qset_propagation}). However, this state is insufficient to defend over two time steps, and thus it is not included in $\nsp{2}{1}$. 
Similarly, we can see that the state $\x=[0,0,2]\in \nsp{2}{1}$ is sufficient to guard over two time steps, but not for four or more time steps.
The vertices of $\nsp{4}{1}$ (top right subfigure) are $[0,0,3]$, $[1,1,1]$, $[1,0,2]$, and $[0,2,1]$. One can verify that any of these states, as well as any convex combination of these states is sufficient to guard against one unit of no-splitting \attacker{} indefinitely.

\subsection{Effect of Edges on CRR}
The relationship between the CRR and the graph structure is not straightforward.
One might, for example, expect a positive correlation between the number of edges and the CRR, since an increase in the number of outgoing edges from a node increases the number of neighboring nodes that must be covered by the \defender{}. 
However, we show by a counter-example (found by the algorithm) that this is not the case.

The following example illustrates how the addition of edges can drastically change the CRR.
\figref{fig:alpha_example} provides examples of directed graphs with five nodes but with different edge sets.
The corresponding indefinite-defense CRR, $\factor{\infty}{}$, for each graph is obtained using \algref{alg:Q-prop}.

In the simplest ring-graph instance, the \defender{} needs only a single robot to indefinitely defend against a single \attacker{}, 
consistent with our prior results reported in~\cite{shishika2022dynamic}.
Interestingly, if the edge between nodes 4 and 5 is made bidirectional, the CRR jumps to $\factor{\infty}{} = 5$, giving the \attacker{} a significant advantage.
If we further add a bidirectional edge between nodes 3 and 4, the CRR decreases to $4$, which benefits the \defender{}.

Finally, if instead of adding the edge between nodes 3 and 4 we introduce a self-loop at node~3, the CRR drops from $5$ to $3$.
This observation highlights that some edges (e.g., the self-loop at node~3) have a larger impact on the game than others (e.g., the edge between nodes 3 and 4).

\begin{figure}[ht]
    \centering
    \includegraphics[width=0.85\textwidth]{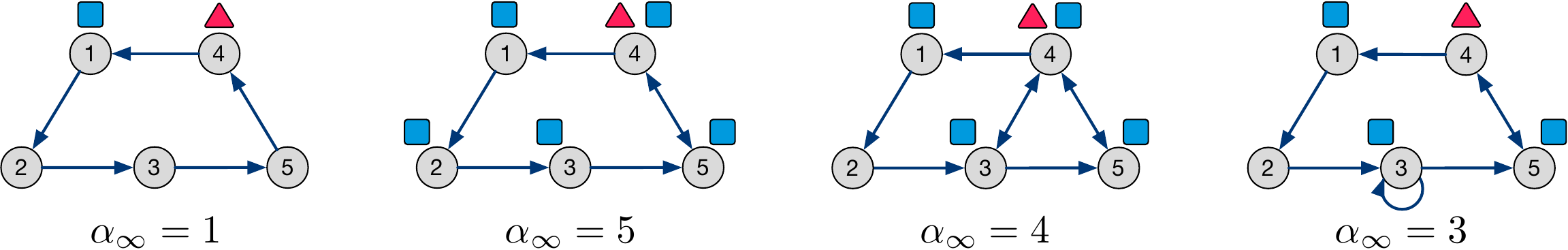}
    \caption{Examples of how CRR changes with the graph structure. All self-loops are explicitly presented. The necessary and sufficient amount of blue agents are placed in the safe set for a given red agent in each figure.}
    \label{fig:alpha_example}
\end{figure}

\subsection{Non-integer Resource Ratio}

Another natural conjecture regarding the CRR is that it must always take integer values.
However, the following dDAB example on a six-node graph (see \figref{fig:35-is-enough}) shows that the critical resource ratio can be non-integer in \textit{finite-horizon} dDAB games.
For this example, \algref{alg:Q-prop} returns $\factor{2}{} = 3.5$. In other words, 3 units of \defender{} resources are insufficient to guarantee a two-step defense against a single unit of \attacker{} resource, whereas 3.5 units are sufficient.
A detailed explanation of why 3 units are insufficient is provided in Appendix~\ref{appdx-sec:non-integer}. Here, we focus on presenting the strategy that allows the \defender{} to successfully defend with 3.5 units of resources.

\begin{figure}[b]
    \centering
    \includegraphics[width=0.2\linewidth]{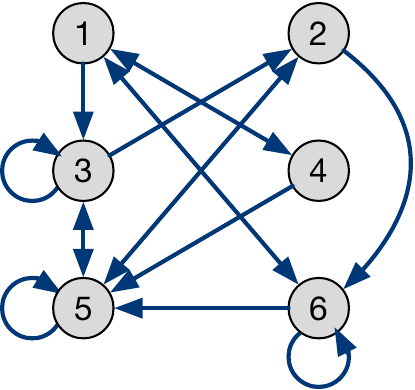}
    \vspace{-5pt}
    \caption{A six-node graph. All self-loops are explicitly presented.}
    \label{fig:35-graph}
\end{figure}

\figref{fig:35-is-enough}
presents the game tree starting with 3.5 units of \defender{} resource, and we show that regardless of the (no-splitting) strategy used by the \attacker{}, the \defender{} can defend until the end of time step 2.
{Since $\factor{2}{}=3.5$ is attained with $\beta_2^{(3)}=3.5$, we let the \attacker{} start with $\yinit=\y^{(3)}$.}
It is easy to verify that the initial \defender{} state $\x_0$ is in $\polyreq(\y^{(3)})$. 
The \attacker{} has three feasible moves at $t=0$: move to node 2, move to node 5, or stay at node 3. 
We only present the first two moves in \figref{fig:35-is-enough}, since for the third move, the \defender{} can just maintain its current state as a countermeasure and does not lose any defense time.\footnote{Even though node 2 does not have a self-loop, the \defender{} resources on nodes 2 and 5 can swap locations to keep the current configuration.}
Furthermore, we focus on explaining the \attacker{}'s move to node~2, since the defense against \attacker{} moving to node~5 can be achieved without using the half unit of resource on node~6.
After observing that the \attacker{} moves to node~2, the \defender{} takes action (a)%
\footnote{Action (a) splits \defender{} resources so that the unit of \defender{} on node~2 moves to node~6; 
the half unit on node~3 moves to node~2 
and the other half stays on node~3; 
the half unit on node~6 moves to node 1, and finally the unit on node 5 stays.} 
and arrives at the state at the beginning of time step~1.
The \attacker{} then has two options, either move to node~5 or to node~6.
Suppose the \attacker{} moves to node~6, the \defender{} initiates action~(b)\footnote{Action~(b) moves the half unit on node~1 to node~5;
the half unit on node~2 to node~5,
the unit on node~6 to node~1, 
and the rest of the resources on nodes 3 and 5 stay.}, which ensures that the configuration at the beginning of time step~2 is still in the required set.  
Similar moves can be made for trajectory (ii) to ensure the defense until the end of time step~2.
For more details regarding the \defender{} actions (a) to (c), see Appendix~\ref{appdx-sec:non-integer}.

By dynamically redistributing fractional resources, the \defender{} achieves defense with only an additional half unit of resource. 
The strategy presented is found by the algorithms in \secref{sec:nsp}, which verifies the efficacy of the proposed approach. 

\begin{figure*}[tb]
    \centering
    \includegraphics[width=0.95\linewidth]{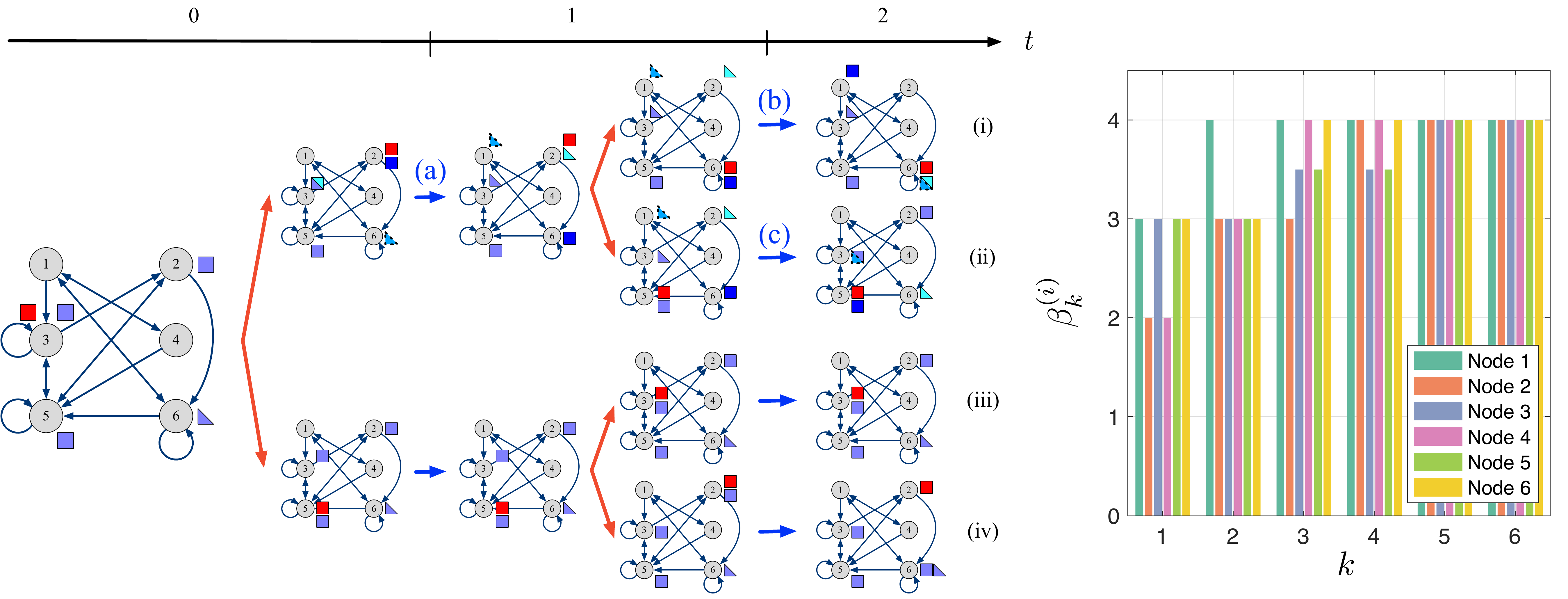}
    \caption{ {(Left) The two time-step game tree over a six-node graph. All self-loops are explicitly presented.  
    The \attacker{} and \defender{} resources are visualized as red and blue boxes, and the blue triangles represent half unit of \defender{} resource.
    Different blue colors are introduced to better visualize the splitting and regrouping of the \defender{} resources. 
    (Right) The evolution of $\nspfactornode{k}{i}$ on each node, until they converge to $\alpha_\infty\!=\!4$ at horizon $k_\infty \!=\! 6$.}
    }
    \label{fig:35-is-enough}
\end{figure*}

\subsection{Experiments on the Robotarium Testbed}

We implement the proposed dDAB algorithm and the resulting \defender{} and \attacker{} strategies on the Robotarium platform ~\cite{wilson2020robotarium}
to demonstrate the deployability on a physical multi-robot system.
While the set-based dynamic program in~\eqref{eqn:Q-set-def} operates with continuous resources, the Robotarium
experiments additionally demonstrate how the same algorithm can be used with discrete, embodied resources (mobile robots) via a simple discrete allocation strategy wrapper.

Specifically, at each time step, the \defender{} first selects a target Q-set $\nsp{k}{j}$ based on the observed \attacker{} allocation, as described in line~6 of Algorithm~7.
Given the current \defender{} \emph{discrete} allocation $\x_t$, the wrapper first computes the intersection $\RSet(\x_t) \cap \nsp{k}{j}$ and then selects a discrete point $\x_{t+1}$ (corresponding to $X$ discrete robots) within the intersection.
Based on the new discrete allocation, each \defender{} robot is assigned a target node to achieve the next discrete allocation $\x_{t+1}$, and the location of the assigned node is used as the robot’s target waypoint.
These waypoints are then sent to the Robotarium control interface of the multi-robot testbed, where the built-in safety barrier certificates ensure collision-free execution.

We evaluate the implementation on two representative examples to highlight different operational scenarios of dDAB.

\paragraph{Scenario 1}
This scenario emulates a broad-area outdoor defense task over a network of seven nodes. 
The five nodes marked with squares are the key nodes that the \defender{} needs to constantly maintain numerical advantage, while the the remaining two circular nodes are the \attacker{}'s spawning nodes. 
The \attacker{} robots may appear from the forest or arrive from the sea at the two circular nodes and attempt to breach the defense at the square (key) nodes. 

The Q-set computation indicates that four \defender{} robots are required to indefinitely hold off a single \attacker{}; accordingly, we first run an experiment with two \attacker{} robots (red) versus eight \defender{} robots (blue). 
Fig.~\ref{fig:robotarium-outdoor-1} present snapshots of the experiment taken at the moments when the \attacker{} is about to select its next allocation. 
It can be observed that the neighboring nodes of those currently occupied by the \attacker{} robots consistently contain a sufficient number of \defender{} robots, ensuring successful defense regardless of the \attacker{}'s next move.

\begin{figure}
    \centering
    \includegraphics[width=\linewidth]{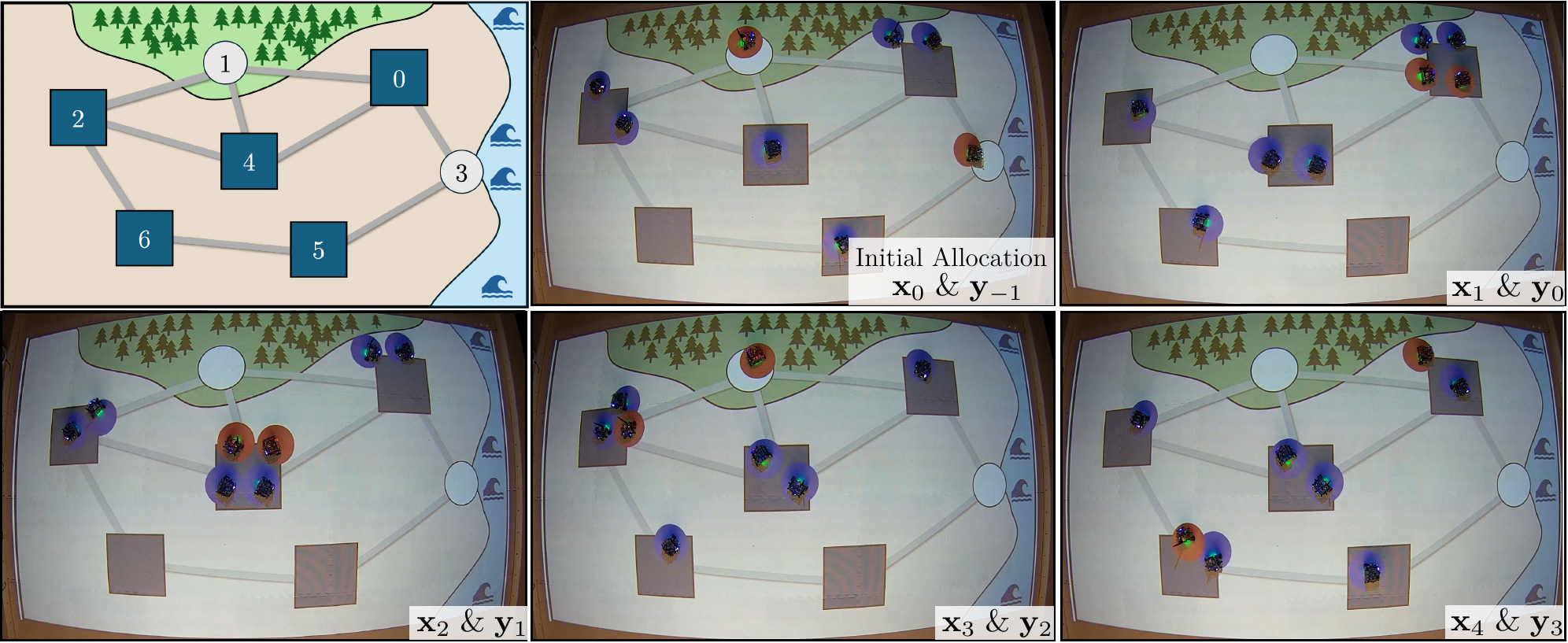}
    \vspace{-0.25in}
    \caption{
    Snapshots from 
    hardware experiment
    for Scenario~1 with eight \defender{} robots against two \attacker{} robots. 
    The first panel shows the underlying graph, and the remaining panels show the discrete allocations at the moments when the \attacker{} selects its next move. 
    The \defender{} team consistently positions sufficient robots at nodes neighboring the \attacker{} robots’ locations, preventing any potential breach after \attacker{}'s next move.}
    \label{fig:robotarium-outdoor-1}
    \vspace{-0.1in}
\end{figure}

Next, we repeat the experiment after removing one \defender{} robot from node~2.
Under this reduced \defender{} team, Algorithm~\ref{alg:attacker-strategy} predicts an earliest breach at $t\!=\!2$ for the \attacker{} robot initiated from node~3. 
Figure~\ref{fig:robotarium-outdoor-2} shows the movement sequence selected by the \attacker{} that leads to this breach at node~2, along with the corresponding \defender{} responses.

\begin{figure}[h]
    \vspace{-0.2in}
    \centering
    \includegraphics[width=\linewidth]{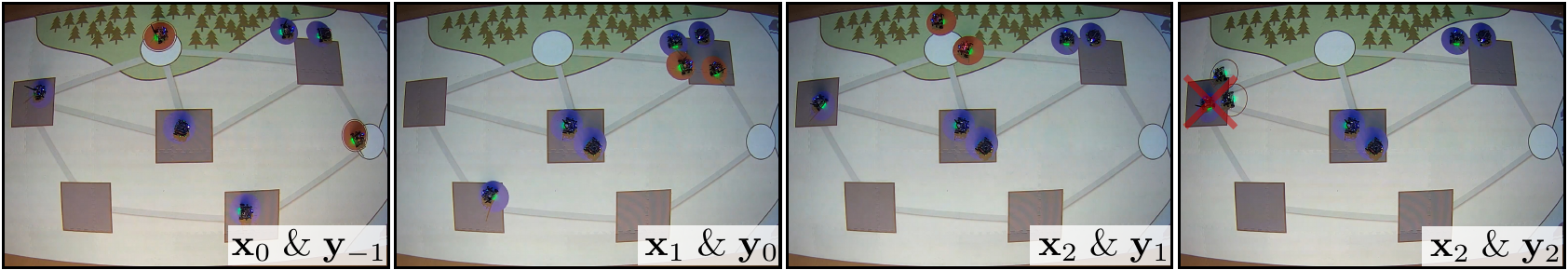}
    \vspace{-0.2in}
    \caption{Allocation sequence for Scenario~1 after removing one \defender{} robot. The \attacker{} robots achieve a breach at node~2 at time step $t\!=\!3$.}
    \label{fig:robotarium-outdoor-2}
    \vspace{-0.15in}
\end{figure}

\paragraph{Scenario 2}
This scenario represents an indoor surveillance problem with nine rooms (nodes~0–8), of which six are key rooms (nodes~0–5). 
The \defender{} team must ensure that, at every time step, at least one defender is present in the same key room as the \attacker{}. 
We deploy four \defender{} robots to defend against a single \attacker{} robot—a configuration that, according to the Q-set analysis, guarantees indefinite defense. 
The \attacker{} is spawned outside the building at node~9, enters through node~6, and then explores the rooms at random. As shown in Fig.~\ref{fig:robotarium-indoor}, at each time step there is always a \defender{} robot co-located with the \attacker{} and at least one \defender{} robot positioned in each neighboring room, thereby maintaining continuous surveillance throughout the experiment regardless of the \attacker{}'s moves.

\begin{figure}[t]
    \centering
    \includegraphics[width=\linewidth]{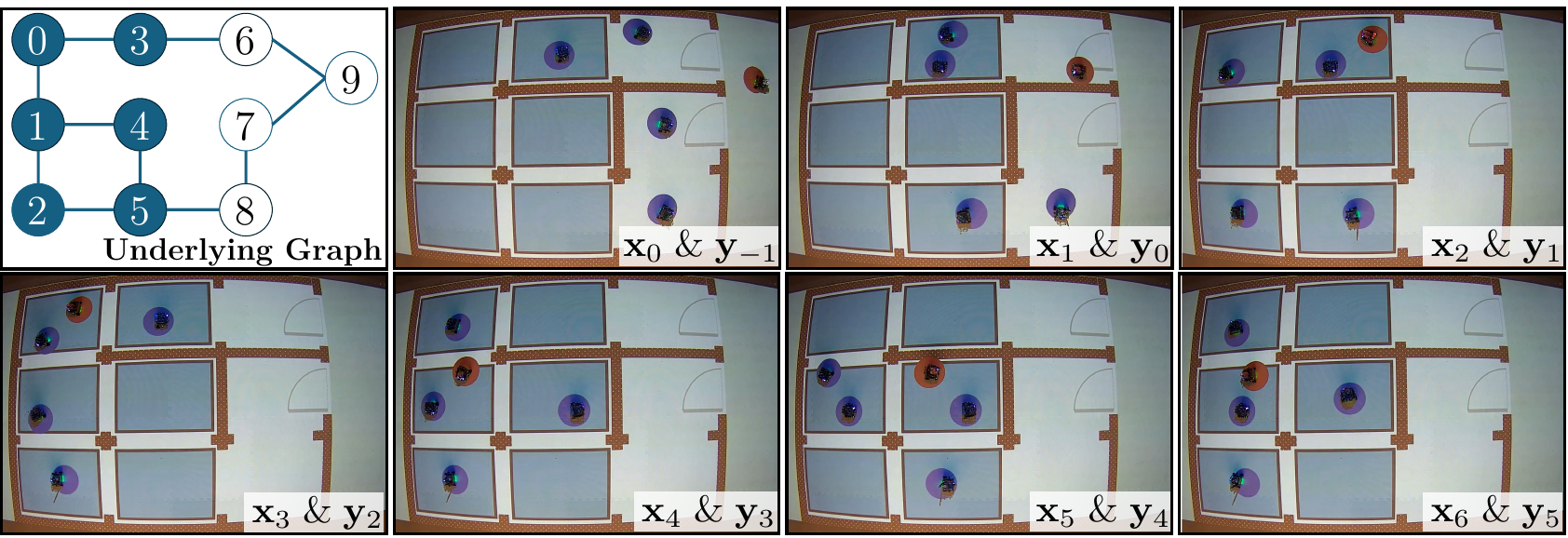}
    \vspace{-0.25in}
    \caption{Snapshots from the indoor surveillance experiment with four \defender{} robots against one \attacker{}. 
    At each time step, the \defender{} robots maintain presence in the \attacker{}’s room and its neighboring rooms, ensuring continuous surveillance.}
    \label{fig:robotarium-indoor}
    \vspace{-0.1in}
\end{figure}

\section{Open Problems}
The Q-prop algorithm in~\algref{alg:Q-prop} is an iterative algorithm that finds the Q-sets. 
Through the sink problem in~\secref{appdx:sink-example}, we have shown that there are graphs on which the Q-prop algorithm does not converge. 
Perhaps there are general conditions on the graph that guarantees that the Q-prop algorithm converges to the indefinite defense Q-sets. 
It is also of interest to see the convergence behavior of the algorithm, i.e. asymptotic vs. finite iteration. 

Empirically, we observed that the critical resource ratio  $\alpha_\infty$ for undirected graph is always integer-valued.
We further observed that $\alpha_\infty \leq |\nodeset|$ for all undirected graphs and $\alpha_\infty > |\nodeset|$ only for directed graphs. 
It is unclear whether these two observations can be formally proved or if additional strengthened conditions on the underlying graphs are required. 

\section{Conclusion}
In this work, we formulated a dynamic adversarial resource-allocation problem by combining the Colonel Blotto game with ideas from population dynamics on graphs. 
Instead of achieving a desired allocation instantly as in traditional Blotto Game formulation, we require that players' resources traverse through the edges of the graph. 
We developed an efficient reachable-set approach to predict the state evolution.
We fully characterize the game by deriving the necessary and sufficient condition (the Q-sets) for either of the player to win the game, along with the corresponding reactive strategies. 
The efficacy of the proposed approach is verified through numerical simulations and physical experiments on the Robotarium platform.
Future work will investigate conditions required for the convergence of the Q-prop algorithm, which leads to guaranteed indefinite defense. 
It is also of interest to consider heterogeneous resources as in~\cite{prorok2017impact} and decentralized decision-making via the common-information approach~\cite{guan2023zero}.

\bibliographystyle{IEEEtran}
\bibliography{refs}

@article{ravichandar2020strata,
  title={{STRATA}: unified framework for task assignments in large teams of heterogeneous agents},
  author={Ravichandar, Harish and Shaw, Kenneth and Chernova, Sonia},
  journal={Autonomous Agents and Multi-Agent Systems},
  volume={34},
  number={2},
  pages={1--25},
  year={2020},
  publisher={Springer}
}

@article{duan2025graph,
  title={Graph Attention Network for Predicting Duration of Large-Scale Power Outages Induced by Natural Disasters},
  author={Duan, Chenghao and Ji, Chuanyi},
  journal={arXiv preprint arXiv:2511.10898},
  year={2025}
}

@article{wilson2020robotarium,
  title={The {R}obotarium: Globally impactful opportunities, challenges, and lessons learned in remote-access, distributed control of multirobot systems},
  author={Wilson, Sean and Glotfelter, Paul and Wang, Li and Mayya, Siddharth and Notomista, Gennaro and Mote, Mark and Egerstedt, Magnus},
  journal={IEEE Control Systems Magazine},
  volume={40},
  number={1},
  pages={26--44},
  year={2020},
  publisher={IEEE}
}

@article{lerman2006analysis,
  title={Analysis of dynamic task allocation in multi-robot systems},
  author={Lerman, Kristina and Jones, Chris and Galstyan, Aram and Matari{\'c}, Maja J},
  journal={The International Journal of Robotics Research},
  volume={25},
  number={3},
  pages={225--241},
  year={2006},
  publisher={SAGE Publications}
}

@inproceedings{guan2023zero,
  title={Zero-Sum Games between Mean-Field Teams: Reachability-based Analysis under Mean-field Sharing},
  author={Guan, Yue and Afshari, Mohammad and Tsiotras, Panagiotis},
  booktitle={Proceedings of the AAAI Conference on Artificial Intelligence},
  pages={15930--15937},
  year={2024}
}

@article{tereshchuk2019efficient,
  title={An efficient scheduling algorithm for multi-robot task allocation in assembling aircraft structures},
  author={Tereshchuk, Veniamin and Stewart, John and Bykov, Nikolay and Pedigo, Samuel and Devasia, Santosh and Banerjee, Ashis G},
  journal={IEEE Robotics and Automation Letters},
  volume={4},
  number={4},
  pages={3844--3851},
  year={2019},
  publisher={IEEE}
}

@inproceedings{shishika2022dynamic,
  title={Dynamic defender-attacker {Blotto} game},
  author={Shishika, Daigo and Guan, Yue and Dorothy, Michael and Kumar, Vijay},
  booktitle={American Control Conference (ACC)},
  pages={4422--4428},
  year={2022},
  address={Atlanta, GA}
}

@ARTICLE{wireless,
  author={Xu, Yongjun and Gui, Guan and Gacanin, Haris and Adachi, Fumiyuki},
  journal={IEEE Communications Surveys \& Tutorials}, 
  title={A Survey on Resource Allocation for 5{G} Heterogeneous Networks: Current Research, Future Trends, and Challenges}, 
  year={2021},
  volume={23},
  number={2},
  pages={668-695},
  doi={10.1109/COMST.2021.3059896}}

@article{nair2018multi,
  title={Multi-agent systems for resource allocation and scheduling in a smart grid},
  author={Nair, Arun Sukumaran and Hossen, Tareq and Campion, Mitch and Selvaraj, Daisy Flora and Goveas, Neena and Kaabouch, Naima and Ranganathan, Prakash},
  journal={Technology and Economics of Smart Grids and Sustainable Energy},
  volume={3},
  number={1},
  pages={1--15},
  year={2018},
  publisher={Springer}
}

@inproceedings{anuradha2014survey,
  title={A survey on resource allocation strategies in cloud computing},
  author={Anuradha, VP and Sumathi, D},
  booktitle={International Conference on Information Communication and Embedded Systems (ICICES2014)},
  pages={1--7},
  year={2014},
  address={Chennai, India}
}

@inproceedings{bei2018algorithms,
  title={Algorithms for trip-vehicle assignment in ride-sharing},
  author={Bei, Xiaohui and Zhang, Shengyu},
  booktitle={Proceedings of the AAAI Conference on Artificial Intelligence},
  volume={32},
  number={1},
  year={2018}
}

@article{khamis2015multi,
  title={Multi-robot task allocation: A review of the state-of-the-art},
  author={Khamis, Alaa and Hussein, Ahmed and Elmogy, Ahmed},
  journal={Cooperative Robots and Sensor Networks 2015},
  pages={31--51},
  year={2015},
  publisher={Springer}
}

@book{bertsimas1997introduction,
  title={Introduction to Linear Optimization},
  author={Bertsimas, Dimitris and Tsitsiklis, John N},
  volume={6},
  year={1997},
  publisher={Athena Scientific Belmont, MA}
}

@article{klumpp2019dynamics,
  title={The dynamics of majoritarian {Blotto} games},
  author={Klumpp, Tilman and Konrad, Kai A and Solomon, Adam},
  journal={Games and Economic Behavior},
  volume={117},
  pages={402--419},
  year={2019},
  publisher={Elsevier}
}

@inproceedings{hajimirsaadeghi2017dynamic,
  title={A dynamic colonel {Blotto} game model for spectrum sharing in wireless networks},
  author={Hajimirsaadeghi, Mohammad and Mandayam, Narayan B},
  booktitle={Annual Allerton Conference on Communication, Control, and Computing (Allerton)},
  pages={287--294},
  year={2017},
  organization={IEEE}
}

@article{konrad2018budget,
  title={Budget and effort choice in sequential Colonel {Blotto} campaigns},
  author={Konrad, Kai A},
  journal={CESifo Economic Studies},
  volume={64},
  number={4},
  pages={555--576},
  year={2018},
  publisher={Oxford University Press}
}

@article{kovenock2018optimal,
  title={The optimal defense of networks of targets},
  author={Kovenock, Dan and Roberson, Brian},
  journal={Economic Inquiry},
  volume={56},
  number={4},
  pages={2195--2211},
  year={2018},
  publisher={Wiley Online Library}
}

@inproceedings{chandan2020showing,
  title={When showing your hand pays off: Announcing strategic intentions in Colonel {Blotto} games},
  author={Chandan, Rahul and Paarporn, Keith and Marden, Jason R},
  booktitle={American Control Conference (ACC)},
  pages={4632--4637},
  year={2020},
  address ={Denver, CO}
}

@article{berman2009optimized,
  title={Optimized stochastic policies for task allocation in swarms of robots},
  author={Berman, Spring and Hal{\'a}sz, Ad{\'a}m and Hsieh, M Ani and Kumar, Vijay},
  journal={IEEE Transactions on Robotics},
  volume={25},
  number={4},
  pages={927--937},
  year={2009},
  publisher={IEEE}
}

@article{korsah2013comprehensive,
  title={A comprehensive taxonomy for multi-robot task allocation},
  author={Korsah, G Ayorkor and Stentz, Anthony and Dias, M Bernardine},
  journal={The International Journal of Robotics Research},
  volume={32},
  number={12},
  pages={1495--1512},
  year={2013},
  publisher={SAGE Publications Sage UK: London, England}
}

@inproceedings{paarporn2019characterizing,
  title={Characterizing the interplay between information and strength in {Blotto} games},
  author={Paarporn, Keith and Chandan, Rahul and Alizadeh, Mahnoosh and Marden, Jason R},
  booktitle={Conference on Decision and Control (CDC)},
  pages={5977--5982},
  year={2019},
  address ={Nice, France}
}

@techreport{gross1950continuous,
  title={A Continuous Colonel {Blotto} Game},
  author={Gross, Oliver and Wagner, Robert},
  year={1950},
  institution={RAND Corporation}
}

@article{roberson2006colonel,
  title={The colonel {Blotto} game},
  author={Roberson, Brian},
  journal={Economic Theory},
  volume={29},
  number={1},
  pages={1--24},
  year={2006},
  publisher={Springer}
}

@article{prorok2017impact,
  title={The impact of diversity on optimal control policies for heterogeneous robot swarms},
  author={Prorok, Amanda and Hsieh, M Ani and Kumar, Vijay},
  journal={IEEE Transactions on Robotics},
  volume={33},
  number={2},
  pages={346--358},
  year={2017},
  publisher={IEEE}
}

@article{julian2019distributed,
  title={Distributed wildfire surveillance with autonomous aircraft using deep reinforcement learning},
  author={Julian, Kyle D and Kochenderfer, Mykel J},
  journal={Journal of Guidance, Control, and Dynamics},
  volume={42},
  number={8},
  pages={1768--1778},
  year={2019},
  publisher={American Institute of Aeronautics and Astronautics}
}

@article{powell2009sequential,
  title={Sequential, nonzero-sum “{B}lotto”: Allocating defensive resources prior to attack},
  author={Powell, Robert},
  journal={Games and Economic Behavior},
  volume={67},
  number={2},
  pages={611--615},
  year={2009},
  publisher={Elsevier}
}

@inproceedings{guan2023adversarial,
  title={On the Adversarial Convex Body Chasing Problem},
  author={Guan, Yue and Pan, Longxu and Shishika, Daigo and Tsiotras, Panagiotis},
  booktitle={2023 American Control Conference (ACC)},
  pages={435--440},
  year={2023},
  address = {San Diego, CA}
}

@article{friedman1993convex,
  title={On convex body chasing},
  author={Friedman, Joel and Linial, Nathan},
  journal={Discrete \& Computational Geometry},
  volume={9},
  number={3},
  pages={293--321},
  year={1993},
  publisher={Springer}
}

\clearpage

\begin{appendices}
\section{Example of Degenerate Case with Non-strongly-connected Graph}
\label{appdx:sink-example}
In Section~\ref{sec:introduction}, we assumed that the graph is strongly connected. 
Namely, for any pair of node $i,j \in \nodeset$, there is a directed path from node $i$ to node $j$.
This assumption is used to avoid the degenerate case, where a subset of the graph is a sink for the \defender{},  
as shown in \figref{fig:sink_example}.

\begin{figure}[!htb]
    \centering
    \includegraphics[width=0.3\linewidth]{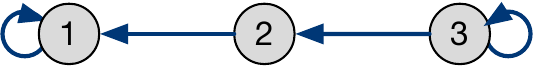}
    \caption{An example of a graph with a sink (node 1). 
    The graph is not strongly connected, since there is no directed path from node 1 to node 3.}
    \label{fig:sink_example}
\end{figure}

One can easily see that by having a single \attacker{} on node 3, the \defender{} must maintain at least one \defender{} resource on node 2. 
This one unit of \defender{} resource will be forced to move to node 1 at the next time step and will stay on node 1 forever. 
Consequently, the \defender{} must ``sacrifice" a unit of its resource at every time step in order to guard node2, and the \attacker{} can trivially win the game by staying on node 3 and wait till the \defender{} runs out of resource and leaves node 2 unattended.

\section{Proof of Theorem~\ref{thm:extreme-action}}
\label{appdx:reachable-sets}

\ea*
\begin{proof}
    We provide a proof by double inclusion.
    The direction of $\convexhull{\hatcalK} \subseteq \K$ is easy to show, as the extreme actions are all admissible actions and the linear constraints in \eqref{eqn:stochastic-matrix-constraint}--\eqref{eqn:adjacency-constraint} hold under convex combinations.
    
    To show that $\K \subseteq \convexhull{\hatcalK}$, we provide a formula of $\{\lambda^{(\ell)}\}_\ell$ in~\eqref{eqn:action-convex-comb} for an arbitrary $K\in \K$. 
    We first define the active edge set $\mathcal{I}^{(\ell)}$ for the extreme action $\hatK^{(\ell)} \in \hatcalK$ as 
    \begin{equation*}
        \mathcal{I}^{(\ell)} = \left\{(j,i) \Big\vert [\hatK^{(\ell)}]_{ij}=1\right\}.
    \end{equation*}
    Then, given any admissible action $K \in \K$, the  coefficients $\lambda^{(\ell)}$ corresponding to the extreme action $\hatK^{(\ell)}$ can be computed as
    \begin{equation}
         \lambda^{(\ell)} = \prod_{(j,i)\in \mathcal{I}^{(\ell)}} [K]_{ij}.
    \end{equation}
    One can further verify that the above formula satisfies~\eqref{eqn:action-convex-comb} and $\sum_{\ell = 1}^{|\hatcalK|} \lambda^{(\ell)} = 1$. 
    Consequently, any admissible action is in the convex hull of the extreme actions.
    With the double inclusion, we have proved the relation in~\eqref{eqn:convex-hull}. 
\end{proof}

\section{Fixed-Point Results}
\label{appdx-sec:fixed-point}

\begin{lemma}
    \label{appdx-lmm:nonempty-intersection}
    Let $A$ be a compact set, $\{B_k\}$ be a sequence of descending closed sets, i.e. $B_0 \supseteq B_1 \supseteq B_2 \supseteq \cdots$. Define $B = \bigcap_{k} B_k$. 
    Suppose that $A\cap B_k \ne \emptyset$, then $A\cap B \ne \emptyset$.
\end{lemma}
\begin{proof}
    Let $x_k \in A\cap B_k$. Since $A$ is compact, there is a convergent subsequence $x_{k_n}$ such that $\lim_{n \to \infty} x_{k_n} = x$ and $x \in A$. We will show that $x\in B$.
    
    Suppose $x \notin B$, then $x \notin B_K$ for some $K$. 
    Furthermore, since $\{B_k\}$ are descending, $x \notin B_k$ for all $k \geq K$.
    Let $N$ be a large enough integer such that $k_N \geq K$.
    Since $B_{K_N}$ is closed, we have that $\mathrm{dist}(x, B_{k_N}) = \min_{y \in B_{K_N}} \norm{x - y} = \epsilon > 0$.
    Furthermore, since $B_{k_N} \supseteq B_{k_{N+1}} \supseteq \cdots$, we have for all $n \geq N$ that
    \[\mathrm{dist}(x, B_{k_n}) \geq \mathrm{dist}(x, B_{k_N}) = \epsilon > 0,\]
    which implies that $\norm{x_{k_n} - x} \geq \epsilon >0$ for all $n \geq N$.
    However, we have that $x_{k_n}\to x$, which is a contradiction. 
    Thus, we have $x \in B$.
    
\end{proof}

\fp*
\begin{proof}
    We provide a proof via double inclusion.  
    Denote $F = \Big\{\x  \big\vert \x \in \polyreq(\y^{(i)}) \wedge \RSet(\x) \cap \nsp{\infty}{\ip} \ne \emptyset ~\forall \ip \in \neighbor{i} \Big\}.$
    
    We first show that $\nsp{\infty}{i} \subseteq F$. 
    Consider an arbitrary $\x \in \nsp{\infty}{i}$. 
    Clearly, $\x \in \nsp{k}{i} \subseteq \polyreq(\y^{(i)})$.
    Furthermore, for all $j \in \neighbor{i}$, we have $\RSet(\x) \cap \nsp{k}{j} \neq \emptyset$ for all $k$.
    Since $\RSet(\x)$ is compact and all $\nsp{k}{j}$ are closed, we can apply Lemma~\ref{appdx-lmm:nonempty-intersection} and conclude that $\RSet(\x) \cap \nsp{\infty}{i} \neq \emptyset$. 
    Consequently, $\nsp{\infty}{i} \subseteq F$.
    
    Next, consider an arbitrary $\x \in F$. Clearly, $\x \in \polyreq(\y^{(i)})$. 
    Furthermore, since $\nsp{\infty}{j} = \bigcap_{k=0}^{\infty} \nsp{k}{j}$, we have $\RSet(\x) \cap \nsp{k}{j} \ne \emptyset$ for all $k$ and $j \in \neighbor{i}$ from the definition of $F$.
    Consequently, $\x \in \nsp{k}{i}$ for all $k$.
    
\end{proof}

\section{Combinition of Subteam Actions}
\label{appdx-1}

\begin{lemma}[Overall action]\label{lmm:overall-action}
Suppose the \defender{} allocation at time $t$ is given as a convex combination:
\begin{equation}
    \x_{t} = \sum_{i} \xi_i \x_{i,t},
\end{equation}
where $\xi_i \geq 0$ and $\sum_{i} \xi_i = 1$. 
Then, for any set of actions $K_{i,t} \in \mathcal{K}$, there exists an admissible overall action $K \in \mathcal{K}$ such that:
\begin{equation}\label{eqn:overall-action-requirement}
    \x_{t+1} = K_t \x_t = \sum_{i} \xi_i K_{i,t} \x_{i,t}.
\end{equation}
Furthermore, the $pq$-th entry of the overall action is given by
\begin{equation} \label{eqn:overal-action-formula}
    \left[K_t\right]_{pq} = 
    \begin{cases}
    \sum_i \left[ \frac{\xi_i \left[\x_{i,t}\right]_q}{\left[\x_t\right]_q} \left[K_{i,t}\right]_{pq} \right] & \text{if } [\x_t]_q > 0, \\
    \frac{1}{d^+_q} &\text{if } [\x_t]_q =0 \text{ and } A_{pq} = 1, \\
    0 & \text{if } [\x_t]_q =0 \text{ and } A_{pq} = 0.
    \end{cases}
\end{equation}
\end{lemma}


\begin{proof}
Let the set $\mathcal{I}=\{i \; | \; [\x_t]_i >0\}$ denote the set of nodes that have non-zero \defender{} resource at time step $t$.
Note that $[\x_t]_q = 0$ implies $[\x_{i,t}]_q = 0$ for all $i$.
From \eqref{eqn:overall-action-requirement}, we have
\begin{align*}
    [\x_{t+1}]_p &= \left[\sum_{i} \xi_i K_{i,t} \x_{i,t}\right]_p\\
    &= \sum_i \xi_i \sum_{q} \left[K_{i,t}\right]_{pq} \left[\x_{i,t}\right]_q\\
    &= \sum_i \xi_i \Bigg( \sum_{q\in \mathcal{I}} \frac{\left[K_{i,t}\right]_{pq} \left[\x_{i,t}\right]_q}{\left[\x_t\right]_q} \left[\x_t\right]_q 
    + \sum_{\substack{q\notin \mathcal{I}\\ A_{pq}=1}} \frac{1}{d^+_q} \underbrace{\left[\x_{i,t}\right]_q }_{=0} 
    + \sum_{\substack{q\notin \mathcal{I}\\ A_{pq}=0}} 0 
    \underbrace{\left[\x_{i,t}\right]_q}_{=0}\Bigg)\\
    &= \sum_{q \in \mathcal{I}} \left[\sum_i  \frac{\xi_i \left[\x_{i,t}\right]_q}{\left[\x_t\right]_q}  \left[K_{i,t}\right]_{pq}\right] \left[\x_t\right]_q + \sum_{\substack{q\notin \mathcal{I}\\ A_{pq}=1}} \frac{1}{d^+_q} \underbrace{\left[\x_{t}\right]_q}_{=0}  
    + \sum_{\substack{q\notin \mathcal{I}\\ A_{pq}=0}} 0 \underbrace{\left[\x_{t}\right]_q}_{=0}\\
    &=\sum_{q} [K_t]_{pq} [\x_t]_q.
\end{align*}
The node $p$ with $[\x_t]_q = 0$ can be ignored, since the relocation action for node $q$ has no effect on $\x_{t+1}$ where there is no resource on node $q$. 
The additional two cases are presented in~\eqref{eqn:overal-action-formula} to solely ensure that the overall action is admissible and well-defined.

Next, we show that the overall action~\eqref{eqn:overal-action-formula} is admissible. 
That is it satisfies the three conditions for admissible actions:
\begin{enumerate}
    \item $K_t \geq 0$ is obvious. 
    \item $[K_t]_{pq} >0$ only if $[A]_{pq}=1$ is also obvious. 
    For the case $[\x_t]_q > 0$   in~\eqref{eqn:overal-action-formula}, $[K_t]_{pq} >0$ only if at least one of $[K_{i,t}]_{pq} >0$. 
    Since $K_{i,t}$ is admissible, its $pq$-th entry can be positive only if $[A]_{pq}=1$.
    For the last two cases where $[\x_t]_q = 0$, the admissibility of $K_t$ is straightforward.
    \item The column sum of $K$ is unity. For the case where $[\x_t]_q \neq 0$, it follows from 
    \begin{align*}
        \sum_p \left[K_t\right]_{pq} 
        &= \sum_p \sum_i \frac{\xi_i \left[\x_{i,t}\right]_q}{\left[\x_t\right]_q} \left[K_{i,t}\right]_{pq}\\
        &=  \sum_i \frac{\xi_i \left[\x_{i,t}\right]_q}{\left[\x_t\right]_q}\sum_p \left[K_{i,t}\right]_{pq}\\
        &= \sum_i \frac{\xi_i \left[\x_{i,t}\right]_q}{\left[\x_t\right]_q} = 1.
    \end{align*}
    When $[\x_t]_q = 0$, we have
    \begin{equation*}
        \sum_p \left[K_t\right]_{pq} = \sum_{p \in \{(q,p)\in \edgeset\}} \frac{1}{d_q^+} = 1. 
    \end{equation*}
\end{enumerate}
In summary, we have shown that there exists an overall action that satisfies \eqref{eqn:overall-action-requirement}, and it is also admissible. 
\end{proof}

\section{Proof of Theorem~\ref{thm:kstep_defense_general_new}}
\label{appdx-sec:general-defender-thm}

\mt*

To rigorously establish the above theorem, we first introduce a more precise formulation that explicitly captures the time-step dependencies.

\begin{theorem}\label{thm:kstep_defense_general}
Given the \attacker{}'s current state $\y_{t-1}$, and the \defender{}'s state can be described as a superposition of the subteams:
\begin{equation}
    \label{eq:safe_state_general}
    \x_t = \sum_{i=1}^{N} \x_{t,T}^{(i)} 
     = \sum_i^n \frac{[\y_{t-1}]_i}{Y} \xhat_{T-t}^{(i)},\;\; \text{where\;}\xhat_{T-t}^{(i)}\in\nsp{T-t}{i}.
\end{equation}
Then, the \defender{} has a strategy to guarantee defense until time step $T$ against any admissible \attacker{} strategy.
\end{theorem}
\begin{proof}
We break the proof into three steps. 
\textbf{Step I:} In Lemma~\ref{lem:safeness}, we show that \eqref{eq:safe_state_general} is a sufficient condition for the \defender{} to defend during the current time step.
\textbf{Step II:} Lemma~\ref{lem:induction} provides a strategy to maintain condition~\eqref{eq:safe_state_general} at the next time step against any admissible \attacker{} strategy.
In other words, for $t\in \{0, \ldots, T-1\}$, if $\x_t$ satisfies \eqref{eq:safe_state_general} for a given $\y_{t-1}$, then for any $\y_{t} \in \RSet(\y_{t-1})$, there is an admissible action $K_t$ such that $\x_{t+1} = K_t \x_t$ satisfies~\eqref{eq:safe_state_general} at $t+1$.
\textbf{Step III:} Based on mathematical induction, condition~\eqref{eq:safe_state_general} is satisfied for all time steps.
Therefore, the defense is guaranteed until time~$T$.
\end{proof}

\begin{lemma}[One-step Safety Guarantee]\label{lem:safeness}
If the \defender{} state $\x_t$ satisfies \eqref{eq:safe_state_general}, then we have $\x_t \in \polyreq(\y_{t-1})$.
In other words, \eqref{eq:safe_state_general} provides sufficiency for the \defender{} to defend the current time step $t$.
\end{lemma}
\begin{proof}
Recalling the definition of \defender{} subteams in~\eqref{eq:defender_subteam}, the condition \eqref{eq:safe_state_general} can be written as
\begin{equation*}
    \x_t = \sum_{i=1}^{N} \frac{[\y_{t-1}]_i}{Y} \xhat_{t}^{(i)}.
\end{equation*}
By definition $\xhat_{t}^{(i)}\in\nsp{T-t}{i}$, which implies $\hat{\x}^{(i)}_{t}\in\polyreq(\y^{(i)})$. 
Therefore for any $F_{t-1}\in\F$ we have
\begin{equation*}
    [\hat{\x}^{(i)}_{t}]_j \geq [F_{t-1}\y^{(i)}]_j = Y [F_{t-1} \e_i]_j,\quad \forall\; j \in \nodeset.
\end{equation*}
Multiplying both sides with $\frac{1}{Y}[\y_{t-1}]_i$, it follows that
\begin{equation*}
    \frac{1}{Y}[\y_{t-1}]_i [\hat{\x}^{(i)}_t]_j \geq  [\y_{t-1}]_i [F_{t-1} \e_i]_j.
\end{equation*}
By taking the sum over $i$, we obtain
\begin{align*}
    [\x_t]_j &= \frac{1}{Y} \sum_{i\in\nodeset} [\y_{t-1}]_i [\hat{\x}^{(i)}_t]_j \geq  \sum_{i\in\nodeset}  [\y_{t-1}]_i [F_{t-1} \e_i]_j \\
    &=  \sum_{i\in\nodeset} \left[  [\y_{t-1}]_i  F_{t-1} \e_i \right]_j 
    =   \Big[ F_{t-1} \sum_{i\in\nodeset} [\y_{t-1}]_i  \e_i \Big]_j \\
    & =   \left[ F_{t-1} \y_{t-1} \right]_j.
\end{align*}
Since the above inequality holds for all $F_{t-1} \in \F$, it follows that
$\x_t \in \polyreq(\y_{t-1})$ (see Remark~\ref{rmk:preq-equivalent}).
Consequently, $\x_t$ can defend the current time step $t$.

\end{proof}

The next lemma shows that the \defender{} can preserve the condition in~\eqref{eq:safe_state_general} against any \attacker{} strategy.

\begin{lemma}[Inductive Condition]\label{lem:induction}
Suppose the \defender{}'s state at time $t$ satisfies
\begin{equation}
    \x_t = \sum_{i=1}^{N} \x_{t,T}^{(i)}.
    \label{eq:inductive_condition}
\end{equation} 
Then, for any \attacker{} action $\y_t \in \RSet(\y_{t-1})$, there exists a \defender{}'s reaction $\x_{t+1} \in \RSet(\x_t)$ such that
\begin{equation}
    \x_{t+1} = \sum_{i=1}^{N} \x_{t+1,T}^{(i)},
    \label{eqn:new_defender_split}
\end{equation}
i.e., the \defender{}'s state at the next time step can also be written as a combination of valid subteams defined in \eqref{eq:defender_subteam}.
\end{lemma}

\begin{proof}
Denote an \attacker{}'s action that takes $\y_{t-1}$ to $\y_t$ as $F_{t-1}$.%
\footnote{
This action may be non-unique as discussed in Section~\ref{subsec:reachable_sets}, but its existence suffices for the purpose of this proof.}
Let $\bs{f}_i$ be the $i$-th column of $F_{t-1}$, i.e., $F_{t-1} = [\bs{f}_1, \bs{f}_2, ..., \bs{f}_N]$, where $\bs{f}_i ^\top \bs{1}=1$ (since $F_{t-1}$ is left stochastic).
We can interpret $\bs{f}_i$ to be the splitting action of the \attacker{} subteam on node $i$ at time $t-1$,
where the fraction of a (possibly empty) subteam on node $i$ relocating to node $j$ is given by $[\bs{f}_i]_j$.

For notational convenience, we drop the second subscript~$T$, when denoting the \defender{} subteams, $\x_{t,T}^{(i)}$.
From Definition~\ref{def:defender_subteam}, we have that the re-scaled $i$-th defender subteam satisfies $\xhat_{t}^{(i)} = (Y \x_{t}^{(i)})/[\y_{t-1}]_i \in \nsp{T-t}{i}.$
From the Q-set definition, we can always construct a satisficing \defender{} action $K^{(i\to j)}$ against a no-splitting \attacker{} moving from node $i$ to~$j$, which guarantees that $\xhat_{t+1}^{(i\to j)} = K^{(i\to j)}\xhat_{t}^{(i)} \in \nsp{T-t-1}{j}.$

Intuitively, the $i$-th \defender{} subteam should react to the splitting of the $i$-th \attacker{} subteam in the following manner. 
First, the $i$-th \defender{} subteam is divided into ``sub-subteams", according to the $i$-th \attacker{} subteam's splitting action $\bs{f}_i$ from the previous time step (see~\figref{fig:splitting-attacker}). 
The $j$-th \defender{} sub-subteam of its $i$-th subteam then counteracts the $j$-th \attacker{} sub-subteam that moves from node $i$ to node $j$.
This counteraction is achieved by the \defender{} sub-subteam applying the action $K^{(i \to j)}$.

Following the intuition above, the $j$-th sub-subteam of the $i$-th \defender{} subteam at time step $t$ has the configuration $[\bs{f}_i]_j \x_t^{(i)} = \frac{[\bs{f}_i]_j [\y_{t-1}]_i}{Y}\xhat^{(i)}_t$,
and it applies the action $K^{(i\to j)}$ to counteract the \attacker{} sub-subteam that moved from node $i$ to node $j$.
The next configuration achieved by this \defender{} sub-subteam is then given by
\begin{equation*}
    \x_{t+1}^{(i\to j)} = \frac{[\bs{f}_i]_j[\y_{t-1}]_i}{Y} K^{(i\to j)} \xhat_{t}^{(i)} = \frac{[\bs{f}_i]_j[\y_{t-1}]_i}{Y}\xhat_{t+1}^{(i\to j)}.
\end{equation*}
Note that $\x_{t+1}^{(i\to j)}$ is only a part of the new $j$-th \defender{} subteam, which originated from the previous $i$-th subteam.

By collecting \defender{} resources originating from different subteams that reacted to the \attacker{} resources that ended up at node $j$ (i.e., $\x_{t+1}^{(i\to j)}$ for $i\in\neighbor{j}$), the new $j$-th \defender{} subteam can be computed as
\begin{equation}
    \x_{t+1}^{(j)} = \sum_{i\in \nodeset} \x_{t+1}^{(i\to j)} = \sum_{i\in \nodeset} \frac{[\bs{f}_i]_j[\y_{t-1}]_i}{Y} \xhat_{t+1}^{(i\to j)}.
\end{equation}

We now verify that this is a valid \defender{} subteam, i.e., it is a state in the corresponding Q-set (scaled by the size of the \attacker{} subteam).
By the definition in~\eqref{eq:defender_subteam}, the rescaled new $j$-th subteam is
\begin{equation}
    \xhat_{t+1}^{(j)} = \frac{Y}{[\y_t]_j} \x_{t+1}^{(j)} = \sum_i \frac{[\bs{f}_i]_j[\y_{t-1}]_i}{[\y_t]_j} \xhat_{t+1}^{(i\to j)}.%
    \footnote{Note that $[\y_t]_j \!= \!0$ implies $\x^{(j)}_{t+1}\!=\!0$. In this case, we can set $\xhat_{t+1}^{(j)}=0$.}
\end{equation}
Noting that $\sum_i [\bs{f}_i]_j[\y_{t-1}]_i = \sum_i [F_{t-1}]_{ji}[\y_{t-1}]_i =  [\y_t]_j$, we see that $\xhat_{t+1}^{(j)}$ is a convex combination of the states $\{\xhat_{t+1}^{(i\to j)}\}_i$.
Since Q-sets are polytopes (Theorem~\ref{thm:Q-polytope}), and also since $\xhat_{t+1}^{(i\to j)} \in \nsp{k-1}{j}$ for all $i\in \nodeset$ by construction, we conclude that $\xhat_{t+1}^{(j)}\in \nsp{k-1}{j}$. 
Thus, the new configuration at time $t+1$ can be written as a superposition of valid subteams.

Finally, since $\x_t = \sum_{i, j} \frac{[\y_{t-1}]_i [\bs{f}_i]_j}{Y}\xhat_t^{(i)}$, 
we can construct the overall \defender{} action $K_t$ that takes $\x_t$ in~\eqref{eq:inductive_condition} to $\x_{t+1}$ in~\eqref{eqn:new_defender_split} based on the sub-subteam actions $K^{(i\to j)}$ (see Lemma~\ref{lmm:overall-action} in Appendix~\ref{appdx-1}),
which completes the proof.

\end{proof}

A minimum working example that illustrates the concepts in the above proof is presented in \figref{fig:subteam_mwe}. 
The readers can use the figure as a roadmap for better understanding of the intuition behind Theorem~\ref{thm:kstep_defense_general}.

\begin{figure}
    \centering
    \includegraphics[width=0.87\linewidth]{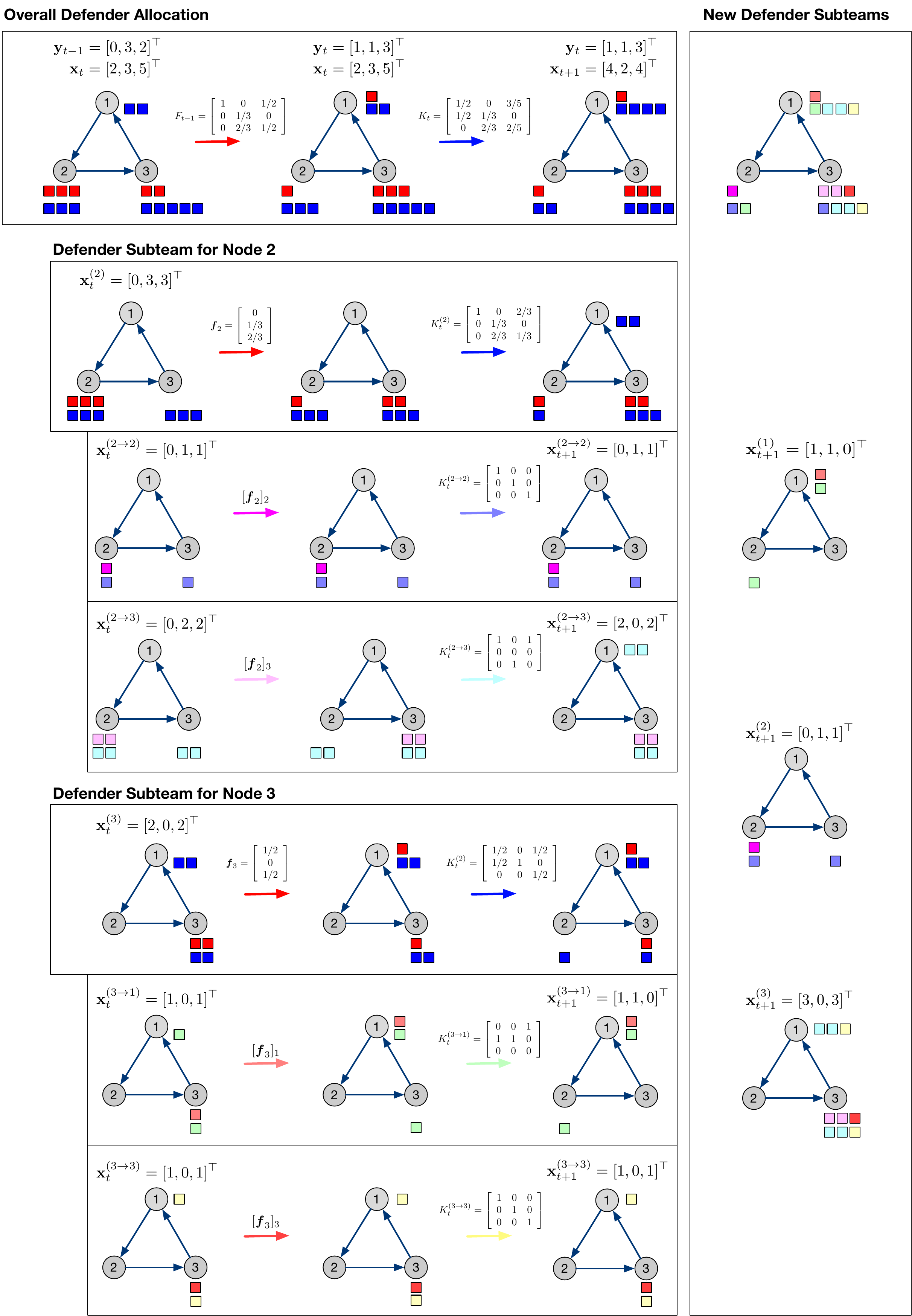}
    \caption{A minimum working example for the proof of Lemma~\ref{lem:induction}.}
    \label{fig:subteam_mwe}
\end{figure}

\newpage
\section{Results on Reversed Graphs}
\label{appdx-sec:reverse-graph}

\ia*
\begin{proof}
When $[\x']_j = 0$, the actions applied on node $j$ has no influence on the next configuration. 
Consequently, we ignore such cases. 
The readers can refer to the proof of Lemma~\ref{lmm:overall-action} for a proper handling of cases where $[\x']_j = 0$. 

We first verify that the reverse action $\widetilde{K}$ in~\eqref{eqn:reverse-action} is admissible. 
It is easy to see that $\widetilde{K}$ satisfies the underlying graph constraints, since the $ij$-th entry of $\widetilde{K}$ depends on the $ji$-th entry of the $K$ matrix.
Each column of $\widetilde{K}$ also sums to unity, since
\begin{equation*}
    \sum_{i} \left[\widetilde{K}\right]_{ij} = \sum_{i} \frac{[K]_{ji} [\x]_i}{[\x']_j} = \frac{\sum_i [K]_{ji} [\x]_i}{[\x']_j} = \frac{[\x']_j}{[\x']_j} = 1.
\end{equation*}

Finally, we show that $\widetilde{K}\x' =\x$.
\begin{align*}
    \left[\widetilde{K}\x'\right]_i &= \sum_{j} \widetilde{K}_{ij} [\x']_j = \sum_{j}\frac{[K]_{ji} [\x]_i}{[\x']_j} [\x']_j\\
    &= \sum_j[K]_{ji}[\x]_i = [\x]_i,
\end{align*}
which completes the proof. 

\end{proof}

\irs*
\begin{proof}
We prove the equality through a double inclusion. 
To show that $\RsetInv(P) \subseteq \widetilde{\RSet}(P)$, select an arbitrary point $\x \in \RsetInv(P)$, then there exists $K \in \K$ and $\x' \in P$, such that $\x' = K \x$. 
According to Lemma~\ref{lmm:inverse-action}, we can construct an inverse action $\widetilde{K} \in \widetilde{\K}$ such that $\x = \widetilde{K} \x'$.
Thus, we have $\x \in \widetilde{\RSet}(P)$.

Consider a point $\x \in \widetilde{\RSet}(P)$. By definition, there exists a point $\x' \in P$ and an admissible action $\widetilde{K} \in \widetilde{\K}$ such that $\x = \widetilde{K} \x'$. 
Note that we can regard the original graph $\graph$ as a reversed graph of its reversed graph $\widetilde{\graph}$.
Consequently, per Lemma~\ref{lmm:inverse-action}, there exists an action $K \in \K$ such that $\x' = K \x$.
Since $\x' \in P$, we have shown that $\x \in \RsetInv(P)$.
\end{proof}

\section{Details of Numerical Examples}
\label{appdx-sec:non-integer}

\subsection{Game Trajectories with Three Units of Resource}
\figref{fig:3-not-enough} presents a game tree where 3 units of \defender{} resource fail to defend against a single \attacker{}.
The \attacker{} selects to start on node 3, i.e. $\yinit=\y^{(3)}$.
The initial \defender{} allocation corresponds to the only feasible state with three unit of \defender{} resource in the $\polyreq(\y^{(3)})$.
The \attacker{} then moves from node 3 to node 2 at time step~0.
Note that with the \attacker{} on node 2, it is necessary for the \defender{} to place one unit of resource on both nodes 5 and 6 to be in the required set, which leads to the three possible configurations at the beginning of time step 1. 
For each of the configurations, the \attacker{} has a corresponding move, which leads to a $\polyreq$ (marked with light blue) that the \defender{} cannot achieve at the beginning of time step 2. 
For example, in trajectory (i), the \attacker{} moves from node 2 to node 5 at time step 1. 
This move leads to a $\polyreq$ that has one unit of \defender{} on each of the nodes 2, 3 and 5, which cannot be achieved by the \defender{}.\footnote{Notice that node 2 does not have a self-loop.}
Consequently, the \attacker{} has a strategy to defeat the \defender{} at the end of time step~2.

\begin{figure}[ht]
    \centering
    \includegraphics[width=0.8\linewidth]{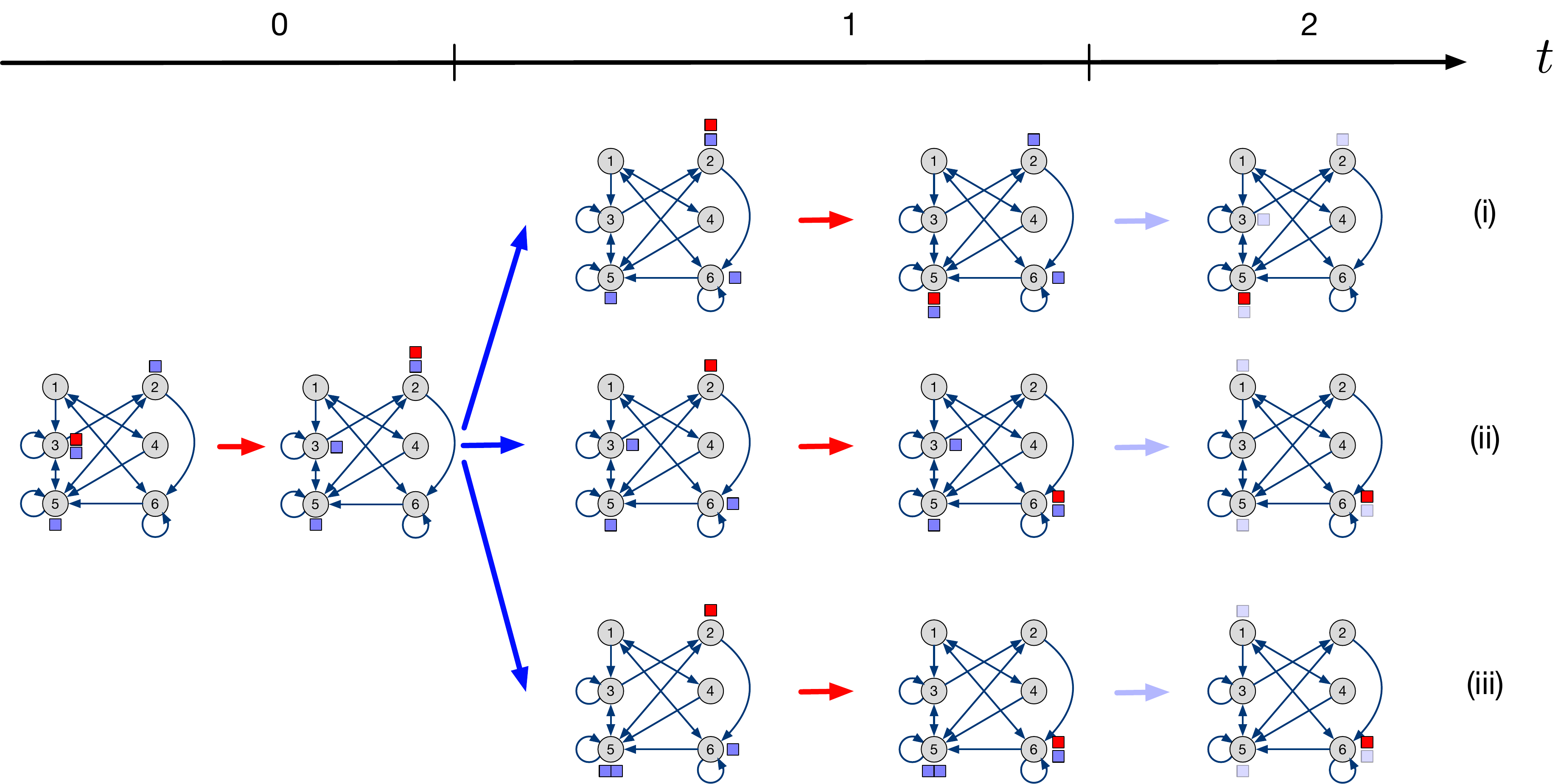}
    \caption{A two-time step game tree starting with one unit of \attacker{} and three unit of \defender{}. Regardless the strategy used by the \defender{}, the \defender{} will be defeated at the end of time step 2.}
    \label{fig:3-not-enough}
\end{figure}

\subsection{Defender Actions Used in Figure~\ref{fig:35-is-enough}}
The \defender{} actions (a) to (c) used in \figref{fig:35-is-enough} are given as follows.
\begin{equation*}
    \text{(a)} = \left[
    \begin{array}{cccccc}
        0 & 0 & 0   & 1 & 0 & 1  \\
        0 & 0 & 0.5 & 0 & 0 & 0  \\
        1 & 0 & 0.5 & 0 & 0 & 0  \\
        0 & 0 & 0   & 0 & 0 & 0  \\
        0 & 0 & 0   & 0 & 1 & 0  \\
        0 & 1 & 0   & 0 & 0 & 0  \\
    \end{array}
    \right], \qquad 
    \text{(b)} = \left[
    \begin{array}{cccccc}
        0 & 0 & 0   & 0 & 0 & 1  \\
        0 & 0 & 0   & 0 & 0 & 0  \\
        0 & 0 & 1   & 0 & 0 & 0  \\
        0 & 0 & 0   & 0 & 0 & 0  \\
        0 & 0 & 0   & 1 & 1 & 0  \\
        1 & 1 & 0   & 0 & 0 & 0  \\
    \end{array}
    \right], \qquad 
    \text{(c)} = \left[
    \begin{array}{cccccc}
        0 & 0 & 0   & 0 & 0 & 0  \\
        0 & 0 & 0   & 0 & 1 & 0  \\
        1 & 0 & 1   & 0 & 0 & 0  \\
        0 & 0 & 0   & 0 & 0 & 0  \\
        0 & 0 & 0   & 1 & 0 & 1  \\
        0 & 1 & 0   & 0 & 0 & 0  \\
    \end{array}
    \right].
\end{equation*}


 



\end{appendices}

\end{document}